\documentclass[11pt,a4paper,reqno]{amsart}
\usepackage{amssymb,amsmath,amsthm,amsfonts,amscd}
\usepackage{graphicx}
\usepackage{mathrsfs}
\usepackage{hyperref}
\hypersetup{colorlinks=true,allcolors=[rgb]{0,0,0.6}}
\usepackage[utf8]{inputenc}
\usepackage[usenames,dvipsnames]{xcolor}
\usepackage[english]{babel}
\usepackage{dsfont}
\usepackage{enumerate}
\usepackage{setspace}
\usepackage{mathabx}
\usepackage{bm}

\usepackage{geometry}
\usepackage[addedmarkup=uwave
]{changes}
\definechangesauthor[name={Marco Falconi}, color={Purple}]{falco}
\definechangesauthor[name={Michele Correggi}, color={red}]{michele}
\definechangesauthor[name={Marco Olivieri}, color={Cyan}]{olivieri}
\setremarkmarkup{\textcolor{Changes@Color#1}{(#2)}}
\usepackage[sort&compress,capitalise,nameinlink]{cleveref}
\crefname{section}{\textsection}{\textsection}
\crefname{subsection}{\textsection}{\textsection}
\crefname{subsubsection}{\textsection}{\textsection}
\crefname{paragraph}{\textparagraph}{\textparagraph}

\renewcommand{\Im}{\mathrm{Im}}
\renewcommand{\Re}{\mathrm{Re}}
          \newtheorem{thm}{Theorem}[section]
          \newtheorem{proposition}[thm]{Proposition}
          \newtheorem{lemma}[thm]{Lemma}
          \newtheorem{corollary}[thm]{Corollary}
          
          \theoremstyle{definition}
          \newtheorem{definition}[thm]{Definition}

          \newtheorem{remark}[thm]{Remark}

\renewcommand{\setminus}{\smallsetminus}

\numberwithin{equation}{section}
\newcommand{\bdm}{\begin{displaymath}}
\newcommand{\edm}{\end{displaymath}}
\newcommand{\bdn}{\begin{eqnarray}}
\newcommand{\edn}{\end{eqnarray}}
\newcommand{\bay}{\begin{array}{c}}
\newcommand{\eay}{\end{array}}
\newcommand{\ben}{\begin{enumerate}}
\newcommand{\een}{\end{enumerate}}
\newcommand{\beq}{\begin{equation}}
\newcommand{\eeq}{\end{equation}}
\newcommand{\beqn}{\begin{eqnarray}}
\newcommand{\eeqn}{\end{eqnarray}}
\newcommand{\bml}[1]{\begin{multline} #1 \end{multline}}
\newcommand{\bmln}[1]{\begin{multline*} #1 \end{multline*}}

\renewcommand{\leq}{\leqslant}
\renewcommand{\geq}{\geqslant}

\newcommand{\disp}{\displaystyle}
\newcommand{\tx}{\textstyle}

\newcommand{\one}{\mathds{1}}
\newcommand{\lf}{\left}
\newcommand{\ri}{\right}
\newcommand{\bra}[1]{\lf\langle #1\ri|}
\newcommand{\ket}[1]{\lf|#1 \ri\rangle}
\newcommand{\braket}[2]{\lf\langle #1|#2 \ri\rangle}
\newcommand{\braketr}[2]{\lf\langle #1\lf|#2\ri. \ri\rangle}
\newcommand{\braketl}[2]{\lf.\lf\langle #1\ri|#2 \ri\rangle}
\newcommand{\meanlr}[3]{\bra{#1}#2\ket{#3}}
\newcommand{\meanlrlr}[3]{\lf\langle #1\lf|#2\ri|#3\ri\rangle}

\newcommand{\xv}{\mathbf{x}}

\newcommand{\kv}{\mathbf{k}}
\newcommand{\kvp}{\mathbf{k}^{\prime}}

\newcommand{\sv}{\mathbf{s}}

\newcommand{\eev}{\mathbf{e}}
\newcommand{\aav}{\mathbf{A}}
\newcommand{\diff}{\mathrm{d}}
\newcommand{\eps}{\varepsilon}

\newcommand{\ceps}{c_{\eps}}

\newcommand{\zv}{\mathbf{z}}
\newcommand{\xxv}{\mathbf{X}}
\newcommand{\kkv}{\mathbf{K}}

\newcommand{\R}{\mathbb{R}}
\newcommand{\M}{\mathscr{M}}

\newcommand{\C}{\mathbb{C}}

\newcommand{\HH}{\mathcal{H}}
\newcommand{\HHe}{\mathcal{H}_{\mathrm{eff}}}
\newcommand{\heps}{\mathcal{H}_{\eps}}
\newcommand{\homega}{\mathfrak{H}_{\omega}}

\newcommand{\hh}{\mathfrak{H}}

\newcommand{\fock}{\Gamma_{\mathrm{s}}(\mathfrak{H})}
\newcommand{\RR}{\mathbb{R}}
\renewcommand{\vec}[1]{\boldsymbol{#1}}

\begin{document}
\onehalfspace
\bibliographystyle{amsalpha}

\title[Magnetic Schr\"{o}dinger Operators as Quasi-Classical Limit of Pauli-Fierz Models]{Magnetic Schr\"{o}dinger Operators as the Quasi-Classical Limit of Pauli-Fierz-type Models}

\author[M. Correggi]{Michele Correggi}\address{Dipartimento di Matematica, ``Sapienza'' Universit\`{a} di Roma, P.le Aldo Moro 5, 00185, Roma, Italia.}\email{michele.correggi@gmail.com}\urladdr{http://ricerca.mat.uniroma3.it/users/correggi/}
\author[M. Falconi]{Marco Falconi}\address{I-Math, Universität Zürich, Winterthurerstrasse 190, CH-8057 Zürich}
\email{marco.falconi@math.uzh.ch}\urladdr{http://user.math.uzh.ch/falconi/}
\author[M. Olivieri]{Marco Olivieri}\address{Dipartimento di Matematica, ``Sapienza'' Universit\`{a} di Roma, P.le Aldo Moro 5, 00185, Roma, Italia.}\email{olivieri@mat.uniroma1.it}

\begin{abstract}
  We study the quasi-classical limit of the Pauli-Fierz model: the system is composed of finitely many non-relativistic charged particles interacting
  with a bosonic radiation field. We trace out the degrees of freedom of the field, and consider the classical limit of the latter. We prove that the
  partial trace of the full Hamiltonian converges, in resolvent sense, to an effective Schr\"{o}dinger operator with magnetic field and a corrective electric
  potential that depends on the field configuration. Furthermore, we prove the convergence of the ground state energy of the microscopic system to the
  infimum over all possible classical field configurations of the ground state energy of the effective Schr\"{o}dinger operator.
\end{abstract}

\keywords{Quasi-classical limit, Magnetic Schr\"{o}dinger Operators, Magnetic Laplacians, Pauli-Fierz Model.}\subjclass[2010]{Primary: 81Q20, 81V10. Secondary: 81T10, 81Q10.}

\maketitle

\section{Introduction and Main Results}

The relevance of Schr\"{o}dinger operators with magnetic fields, also known as {\it magnetic Schr\"{o}dinger operators} (MSO) in modern
Quantum Mechanics is testified by the amount of work in the mathematical physics community on the topic, see
\cite{MR3496663,MR2228679,MR3615042,MR3603960,MR3616340}, to mention just a few among the most recent ones. The main reason for this
interest is mainly related to the role played by MSO in several key phenomena of solid state and condensed matter physics. The presence
of a magnetic field is for instance a necessary ingredient for both the integer and fractional Quantum Hall effects (see, {\it
  e.g.}, \cite{MR1341838,MR885569} for an introduction to Quantum Hall effects, from a mathematical standpoint). Periodic or quasi-periodic
Schr\"{o}dinger operators with possibly slowly varying magnetic fields are involved in the most relevant models for crystals and the
motion of electrons therein \cite{Panati2003}. MSO are involved in the modelling of the response of a
superconductor to an external magnetic field, in particular for very intense fields, and the ultimate loss of superconductivity when
the field penetrates the sample entirely
\cite{doi:10.1142/S0129055X17500052,Correggi2016,fournais2010spectral,sandier2006vortices}. The structure of Landau levels, {\it i.e.},
eigenvalues of MSO, and the restriction to the lowest Landau level is responsible for peculiar quantum phenomena
\cite{doi:10.1142/S0129055X07002900,Rougerie2015}. Singular magnetic fields, as the Aharonov-Bohm fluxes, even appear in the effective
description of two-dimensional particles obeying to fractional statistics \cite{Lundholm2013}, giving rise to even more peculiar
nonlinear effective models \cite{CLR2017}.

From the mathematical view point, the models related to the phenomena mentioned above typically describe $ N $ non-relativistic quantum particles with Hamiltonian\footnote{For the sake of simplicity, we set $ c, \hbar = 1 $ and we assume that all the masses are equal to $ 1/2$, while all the charges equal $ - 1 $.}
\beq
	\label{eq:MSO}
	\sum_{j = 1}^N \lf( - i \nabla_j - \aav(\xv_j) \ri)^2 + V(\xv_1, \ldots, \xv_N),
\eeq
where $ \aav $ is a classical vector potential, with magnetic field
\beq
	\mathbf{B}(\xv) = \nabla \times \aav(\xv),
\eeq
and $ \xv, \xv_j \in \R^3 $. The potential $ V $ contains the interaction among the particles as well as an external trapping, if present.

In spite of their importance, however, MSO should be viewed as effective non-relativistic models where the quantized field is assumed to be classical. A more fundamental model is indeed obtained by coupling the non-relativistic particles to a radiation electro-magnetic field, so obtaining a non-relativistic model of quantum electrodynamics. Whether such a connection could be made rigorous in the appropriate classical limit is the question that we address in this article. We can anticipate that the answer is positive, if a {\it quasi-classical limit} is taken along the lines of \cite{Correggi:2017aa}, {\it i.e.}, if the degrees of freedom of the quantized field can be suitably approximated by their classical counterpart.

In order to state our main results we have however to specify what model of quantum electrodynamics we choose and, because of its generality and simplicity, we focus on the class of {\it Pauli-Fierz (PF) models}, or, more generally, systems described by PF-type Hamiltonians: they are described by a
minimal coupling between a conserved number of non-relativistic quantum particles and a quantized photon field (in Coulomb gauge) and they were introduced, as the name suggests, in the pioneering work \cite{pauli1938theorie}. Concretely, a PF-type Hamiltonian acts on the tensor space
\beq
	\label{eq:space}
	\mathscr{H} = L^2(\R^{3N}) \otimes \Gamma_{\mathrm{s}}(\mathfrak{H}),
\eeq 
where the first factor is associated to the particle degrees of freedom, while $ \Gamma_s(\mathfrak{H}) $ is the bosonic Fock space with one-particle Hilbert space $ \mathfrak{H} = \mathfrak{h}\otimes\C^2 $, given by two copies of $ \mathfrak{h} $. For the sake of simplicity, we have considered here $ N $ spinless particles and chosen a three-dimensional setting, although our results apply to other dimensions too; also the particles might have spin and thus satisfy suitable symmetry constrains  (see below). A PF operator is then formally very similar to \eqref{eq:MSO}, {\it i.e.},
\begin{equation}
	\label{eq:PFham}
	H = \sum_{j=1}^N \left(-i \nabla_{j} - \vec{\varphi}(\xv_j)\right)^2 + V(\xv_1, \ldots, \xv_N) + \diff \Gamma(|\kv|),
\end{equation}
where the major difference is that now $  \vec{\varphi} $ stands for an operator acting on $ \Gamma_{\mathrm{s}}(\mathfrak{H}) $. In order to give the precise expression of $  \vec{\varphi} $ in the Coulomb gauge, one has to introduce the polarization vectors: let $ \eev_{1}(\kv), \eev_2(\kv) \in \R^3 $ be two three-dimensional unit vectors, such that for a.e. $ \kv \in \R^3 $, 
\bdm
	\kv \cdot \eev_{\gamma}(\kv) = 0,	\qquad		\gamma = 1,2,
\edm
and $ \eev_{1} \cdot \eev_{2} = 0 $, {\it i.e.}, $ \hat{\kv}, \eev_1, \eev_2 $ form a basis in $ \R^3 $. Then, the operator $  \vec{\varphi} $ is given by
\beq
	 \vec{\varphi}(\xv) = \sum_{\gamma = 1}^{2} \int_{\RR^3} \diff \kv \: \left(\lambda_{\gamma}^*(\kv) e^{i \kv \cdot \xv} a_{\gamma}(\kv)+\lambda_{\gamma}(\kv) e^{-i \kv \cdot \xv} a_{\gamma}^{\dagger}(\kv)\right) \vec{e}_{\gamma}(\kv),	\label{eq:aav}
\eeq
where $ \lambda(\kv) \in \mathfrak{H} $ is the particle form factor. The field energy $ \diff \Gamma(\omega) $ is the second quantization of a linear dispersion relation $ \omega(\kv) = |\kv| $, {\it i.e.},
\beq
	\label{eq:field energy}
	\diff \Gamma(|\kv|) = \sum_{\gamma = 1}^2 \int_{\R^3} \diff \kv \: |\kv| \; a_{\gamma}^{\dagger}(\kv) a_{\gamma}(\kv),
\eeq
with $ a_{\gamma}, a_{\gamma}^{\dagger} $ the usual creation and annihilation operators satisfying the canonical commutation relations
\beq
	\label{eq:ccr}
	\lf[ a_{\gamma}(\kv), a^{\dagger}_{\gamma'}(\kvp) \ri] = \eps \: \delta_{\gamma, \gamma'} \: \delta(\kv - \kvp).
\eeq

As explained in detail in the monograph \cite{MR2097788}, the Pauli-Fierz operators can be obtained as the quantization of the Abraham
model of extended non-relativistic classical charges coupled to the classical electromagnetic field (see
\cite{2013arXiv1311.1675F
} for a well-posedness result on the Abraham model suitable for quantization and semiclassical analysis). Combining these quantum and
classical descriptions it is possible to cover a wide range of physical phenomena. In this respect, MSO describe a physical situation
that is intermediate between the PF and Abraham models, in which the quantum nature of the particles is preserved, while the field is
macroscopic and therefore it behaves classically. This work provides exactly a bridge between the completely quantum PF regime and the
intermediate quantum-classical regime of MSO models.

We now provide more details about the {\it quasi-classical} regime we want to study: the physical situation we are describing is the one of quantum particles interacting with a very intense radiation field. The average number of field carriers, {\it e.g.}, photons or phonons, is thus very large and much larger than the commutator between the annihilation and creation operators. From the mathematical view point, this is perfectly equivalent to rescale $ a $ and $ a^{\dagger} $ (and therefore the number operator), so that the commutator $ [a, a^{\dagger}] $ is order $ \eps $ and the average number of carriers of order 1. In \eqref{eq:ccr} and in the rest of the paper, we take this point of view and investigate the quasi-classical limit $ \eps \to 0 $ of the PF model.

We anticipate that our main result is that the partial trace of the PF Hamiltonian converges in norm resolvent sense as $ \eps \to 0 $ to an effective operator which is a MSO. The magnetic field (and an additional unexpected electric field) of such a MSO are expressed in terms of the classical state of the quantized field. Moreover, {\it any} reasonable MSO can be obtained as the quasi-classical limit of a suitable PF-type Hamiltonian, although in this case the convergence is a bit weaker (strong resolvent sense). Finally, we show convergence of the corresponding ground state energies.

From the mathematical point of view, the quasi-classical limit of the PF Hamiltonian is much more subtle than the one of other models of particle-radiation interaction. Indeed, the coupling between the two subsystems in \eqref{eq:PFham} is quadratic in the field operator instead of linear as, {\it e.g.}, in the Nelson-type operators or the polaron model (see \cite{Correggi:2017aa}). In addition, the coupling itself involves a non-multiplicative operator acting on the particle subsystem, whose differential part is the gradient. This poses new technical challenges in the control of the convergence of the terms involving the field operators along the classical limit. Finally, when considering the ground state energy convergence, the massless dispersion relation of the PF Hamiltonian is responsible for a convergence to more general classical states (cylindrical Wigner measures) than in the massive case and a completely new machinery \cite{Falconi:2017aa} has to be applied to attack the problem. 

The plan of the paper is the following: the rest of this Sect. is devoted to a detailed description of the model under investigation and the statement of our main results. The proofs are contained in Sect. \cref{sec:convergence,sec:ground state}. Two Appendices are devoted to a brief summary of the tools of infinite dimensional semiclassical analysis used throughout the paper and a discussion of the self-adjointness of the PF operator respectively.

\medskip

\noindent
{\small
{\bf Acknowledgements.} The authors acknowledge the support of MIUR through the FIR grant 2013 ``Condensed Matter in Mathematical Physics (Cond-Math)'' (code RBFR13WAET). M.F. and M.O. are partially supported by the National Group of Mathematical Physics (GNFM--INdAM).}

\subsection{The microscopic model: Pauli-Fierz-type operators}
\label{sec:PF}

Instead of focusing on the three-dimensional case, and in order to be as much general as possible, we consider a PF-type model in $ d $ dimensions, describing a quantum system composed by $N$ non-relativistic extended charges of the same mass (equal to $ 1/2 $) and spin, interacting with an electromagnetic
radiation field.  The request that the particles have the same mass and spin is only for the sake of presentation, and could be avoided. Moreover, we do not take into account the particle statistics, but this can easily implemented in our analysis by suitably restricting the particle Hilbert space. We shall
partly borrow some notation from \cite{Matte:2017ab
}.

The Hilbert space of the full system is
$$  \mathscr{H} = L^2\lf(\Lambda^N; \mathbb{C}^s\ri) \otimes \Gamma_{\mathrm{s}}(\mathfrak{H}), $$
where $L^2(\Lambda^N; \mathbb{C}^s)$, $\Lambda\subset \mathbb{R}^{d}$ open with smooth boundary, is the Hilbert space associated to $ N $ particles with spin
$\frac{s-1}{2}$ and moving inside the same space region $\Lambda$, which might as well coincide with the whole $ \R^d $; $\Gamma_{\mathrm{s}}(\mathfrak{H})$ is the Hilbert space corresponding to the radiation degrees of freedom. Let us recall
that
$$ \Gamma_{\mathrm{s}}(\mathfrak{H}) = \bigoplus_{n=0}^{\infty} \mathfrak{H}^{\otimes_{\mathrm{s}}n}\; ,$$
with $\mathfrak{H}$ the (separable) one-photon Hilbert space, that we assume to be of the form
\beq
	\label{eq:one-particle space}
	\mathfrak{H} = \C^{d-1} \otimes L^2(\mathcal{K},\diff\nu)
\eeq
 for some $\sigma$-finite and locally
compact measurable space $(\mathcal{K},\Sigma,\nu)$ (momentum space). The tensorization by $ \C^{d-1} $ takes into account all the possible polarizations. The standard creation and annihilation operators $a^{\dagger}(\,\cdot \,)$ and $a(\,\cdot \,)$
satisfy the canonical commutation relations as in \eqref{eq:ccr}, i.e, 
\begin{equation}
	\label{eq:ccr2}
 	\lf[ a(\mathbf{f}), a^{\dagger}(\mathbf{g}) \ri] = \eps \, \braket{\mathbf{f}}{\mathbf{g}}_{\mathfrak{H}} = \varepsilon\sum_{\gamma=1}^{d-1} \braket{f_{\gamma}}{g_{\gamma}}_{L^2(\mathcal{K},\nu)}.
\end{equation}
Notice that in our convention $ a(\mathbf{f}) = \lf( a_{1}(f_1), \ldots, a_{d-1}(f_{d-1}) \ri) $ is in fact a vector of $ d-1 $ operators
acting on the degrees of freedom corresponding to different polarizations. In addition, we shall need to consider vector valued
functions for each polarization: given functions $\mathbf{F},\mathbf{G}\in \mathbb{R}^d\otimes \mathfrak{H}$, we define the commutator
\begin{equation}
	\label{eq:cc3}
 	\lf[ a(\mathbf{F}), a^{\dagger}(\mathbf{G}) \ri] = \eps \, \sum_{i=1}^d\braket{\mathbf{F}_{i}}{\mathbf{G}_{i}}_{\mathfrak{H}} = \varepsilon \sum_{i=1}^d\sum_{\gamma=1}^{d-1} \braket{F_{i,\gamma}}{G_{i,\gamma}}_{L^2(\mathcal{K},\nu)}.
\end{equation}
Therefore in this case the creation/annihilation operators $a^{\dagger}(\mathbf{F})$ and $a(\mathbf{G})$ have $d(d-1)$
components\footnote{Throughout the paper, we denote by bold-face letters vectors with either $d$, $d-1$, $d(d-1)$ $dN$, $d(d-1)N$
  components. It should always be clear from the context which is the case.}. As we have already commented on, $\eps \ll 1 $ plays here the role of a scale parameter, measuring the non-classical behavior of the field part of the system. This very same role is played by $ \hbar $, which is however set equal to 1 here, and therefore the limit $ \eps \to 0 $ should be interpreted as a suitable scaling limit.

The energy of the field is the second quantization of a dispersion relation $\omega:\mathcal{K}\to \mathbb{R}$, that we suppose $\nu$-almost everywhere strictly positive, in a way such
that
\beq
	\label{eq:dgamma}
	\mathrm{d}\Gamma(\omega) = \sum_{\gamma = 1}^{d-1} \int_{\mathcal{K}}\mathrm{d}\nu(\kv) \:\omega(\kv) \; a_{\gamma}^{\dagger}(\kv) a_{\gamma}(\kv)
\eeq
is a positive densely defined self-adjoint operator on the Fock space
$\Gamma_{\mathrm{s}}(\mathfrak{H})$, and $\omega^{-1}$ is a (possibly unbounded) densely defined self-adjoint multiplication operator on
$\mathfrak{H}$. The case $ \omega(\kv) = |\kv| $ is then recovered as a special case.
 
 In order to complete the definition of the PF model, we also have to introduce the form factors of the particles. Let $ \mathscr{D}[\omega + \omega^{-1}] \subset \mathfrak{H} $ be the domain of the closed quadratic form associated to the multiplication operator $ \omega(\kv) + \omega^{-1}(\kv) $ on $ \mathfrak{H} $. We pick a vector-valued function 
\beq
	\label{eq:lam} 
	\bm{\lambda} \in L^{\infty}\lf(\Lambda^N; \lf(\mathscr{D}[\omega + \omega^{-1}]\ri)^{dN} \ri)
\eeq
of the form
\begin{equation}
	\label{eq:lambda}
  	\vec{\lambda}(\xv_1, \ldots, \xv_N; \kv) =\Big( \lambda_{1}(\xv_1; \kv) \eev_{1}(\kv), \lambda_{1}(\xv_1; \kv) \eev_{2}(\kv), \ldots, \lambda_{j}(\xv_j; \kv) \eev_{\gamma}(\kv), \ldots \Big)
\end{equation}
where $ \lambda_{j}(\: \cdot \:\:\: ; \: \cdot \:) \in L^{\infty}(\Lambda;L^2(\mathcal{K},\nu)) $, {\it i.e.}, for a.e. $ \xv_j \in \Lambda $, $ \lambda_{j}(\xv_j;\: \cdot \: ) \in L^2(\mathcal{K},\nu) $, and $ \eev_{\gamma}(\kv) \in \R^d $, $ \gamma = 1,\ldots, d-1 $, are the polarization vectors, {\it i.e.}, $ \hat{\kv}, \eev_{1}(\kv),\ldots, \eev_{d-1}(\kv) $ form an orthonormal basis in $ \R^d $ for any $ \kv $. Physically speaking, $ \lambda_{j}$ is the Fourier transform of the (signed) charge distribution $ \rho_j: \hat{\mathcal{K}}\to \mathbb{R}$ of the
$j$-th particle ($\hat{\mathcal{K}}\cong \Lambda$ is the configuration space), multiplied by the factor $\omega^{-1/2}(\kv) e^{-i \kv \cdot \xv_j}$. Notice that we have implicitly chosen the Coulomb gauge for the field, which is apparent in the request $ \kv \cdot \mathbf{e}_{\gamma} = 0 $, for any $ \gamma = 1, \ldots, d-1 $.  Since we are mostly interested in the Coulomb gauge for the magnetic potential $\varphi(\vec{\lambda})$, we assume for the sake of simplicity that for a.e. $ \xv_j \in \Lambda$, $ j = 1, \ldots, N $ and for all $\gamma=1,\dotsc,d-1$,
\begin{equation}
  	\label{coulomb}\tag{A1}
  	\nabla_j  \lambda_{j}(\xv_j; \kv) \cdot \eev_{\gamma}(\kv) = 0.
\end{equation}

To take into account also a possible coupling between the spin and the magnetic field, which classically would have the usual form $ \sigma \cdot \mathbf{B} $, we introduce another coupling factor 
\beq
	\label{eq:b1}
	\mathbf{b} \in L^{\infty} \lf(\Lambda^N; \lf(\mathscr{D}[\omega^{-1}]\ri)^{dN} \ri) 
\eeq
of the form
\begin{equation}
	\label{eq:b2}
  	\mathbf{b}(\xv_1, \ldots, \xv_N; \kv) =\Big( b_{1}(\xv_1; \kv) \eev_{1}(\kv), b_{1}(\xv_1; \kv) \eev_{2}(\kv), \ldots, b_{j}(\xv_j; \kv) \eev_{\gamma}(\kv), \ldots \Big).
\end{equation}
 The physical meaning of this coupling is to provide the field coupled to the spin, {\it e.g.}, in $ d = 3 $ it is the vector of components $ \nabla_{\xv_j} \times \lambda_{j}(\xv_j; \kv) $.

The PF Hamiltonian of the
microscopic system is thus (the quadratic form associated to)
\begin{equation}
  \label{eqn:ham}
  \framebox{$ H = \lf[ \bigl(-i\nabla - \bm{\varphi}(\vec{\lambda})\bigr)^2 +V + \mathrm{d}\Gamma(\omega) \ri] \otimes \mathds{1}_{\C^s} -\vec{\sigma}\cdot \bm{\varphi}(\mathbf{b})\;, $}
\end{equation}
where 
\begin{equation}
\begin{split}
	\label{eq:field}
	\bm{\varphi}(\bm{\lambda})&=a^{\dagger}(\bm{\lambda})+a(\bm{\lambda}) = \Big( \vec{\varphi}_{1,1}(\lambda_1), \ldots, \vec{\varphi}_{1,d-1}(\lambda_1), \ldots, \vec{\varphi}_{j,\gamma}(\lambda_j), \ldots \Big)	\\
	\vec\varphi_{j,\gamma}(\lambda_j(\xv_j)) &= a^{\dagger}_{\gamma}\lf(\lambda_j(\xv_j) \eev_{\gamma} \ri)+a_{\gamma}(\lambda_j(\xv_j)  \eev_{\gamma})	\\
									& = \int_{\mathcal{K}} \diff \nu(\kv) \:  \lf[ \lambda_{j}(\xv_j; \kv) \mathbf{e}_{\gamma}(\kv) a^{\dagger}_{\gamma}(\kv) + \lambda^*_{j}(\xv_j; \kv) \mathbf{e}_{\gamma}(\kv) a_{\gamma}(\kv) \ri]
\end{split}
\eeq
is the usual Segal field, interpreted as the quantum vector potential; $V=V(\xv_1, \cdots, \xv_N)\otimes \mathds{1}_{\Gamma}$ is the self-interaction potential among the particles; and $\vec{\sigma}=(\vec{\sigma}_1,\dotsc,\vec{\sigma}_{N})$, $ \vec{\sigma}_j = (\sigma_{j,1}, \ldots, \sigma_{j,d}) $,
is a vector of $s\times s$ spin matrices. The latter is the only operator that acts on the spin degrees of freedom non-trivially, coupled to the quantum magnetic field $\varphi(\mathbf{b})$ (with components of the form \eqref{eq:field}). In absence of spin, {\it i.e.}, for $s=1$, this last term should be dropped from the operator. The interaction potential $V(\xv_1, \ldots, \xv_N)$ is assumed to be of the form
\beq
  	\label{eqn:h3} \tag{A2}
	\begin{split}
  		&V(\xv_1, \ldots, \xv_N) = V_{<}(\xv_1, \ldots, \xv_N) + V_+(\xv_1, \ldots, \xv_N)\; ,  \\
  		&V_{<} \in \mathfrak{K}_{<}\lf(\Lambda^N\ri), \; V_+ \in L^1_{\mathrm{loc}}\lf(\Lambda^N; \RR^+\ri)\;, 
\end{split}
\eeq
where $\mathfrak{K}_{<}$ stands for the set of operators which are Kato-small as quadratic forms (in the sense of KLMN Theorem \cite[Theorem X.17]{MR0493420})\footnote{Here and in the rest of the paper we denote by $ C > 0 $ a generic finite constant whose value may change from line to line.}:
\bmln{
	\mathfrak{K}_{<}(\Lambda) = \Bigl\{ U: \Lambda^N \rightarrow \RR \; \Big| \; \meanlr{\psi}{U}{\psi}  \leq a \lf\| \nabla \psi \ri\|_2^2 + C \lf\| \psi \ri\|_2^2, \forall \psi \in H^1_0\lf(\Lambda^N\ri),	\\
	 \text{ for some } a < 1 \text{ and } C < +\infty  \Bigr\}.
} 
With the above assumptions, $H$ is self-adjoint and bounded from below on $ \mathscr{D}(-\Delta_{\mathrm{D}}+V_+)\otimes \mathscr{D}(\mathrm{d}\Gamma(\omega)) $ (see \cref{sec:PFO}).

\subsection{Quasi-classical limit}
\label{sec:conveff}

In order to study the effective behavior of the particles as $\varepsilon\to 0$, we need to consider suitable $\varepsilon$-dependent Fock vectors, and then trace out the
radiation  degrees of freedom with respect to such vectors. The aim is to prove that the resulting effective operator on the particle space
$L^2(\Lambda^N;\mathbb{C}^s)$ has a limit (in suitable sense) as $\varepsilon\to 0$, and that we can characterize such a limit as a MSO of the form \eqref{eq:MSO}.

For any state of the field represented by a normalized $\varepsilon$-dependent Fock vector $ \Psi_{\eps} \in \Gamma_{\mathrm{s}}(\mathfrak{H})$, one can find a cylindrical measure $\mu\in \mathscr{M}_{\mathrm{cyl}}(\mathfrak{H})$ and a sequence $\Psi_{\varepsilon_n}$, with $\varepsilon_n\in (0,1)$ for all $n\in \mathbb{N}$, such that in a suitable topology $\Psi_{\varepsilon_n}$ converges to $\mu$. To different sequences may correspond different limits. It is convenient to adopt the following shorthand notation: let $\mu$ be one cluster point of $\Psi_{\varepsilon}$ in the aforementioned topology, we denote by
\beq
	\label{eq:conv}
	\Psi_{\eps} \xrightarrow[\eps \to 0]{} \mu,
\eeq
the convergence of $ \Psi_{\eps} $ to $\mu$ along all suitable sequences. The corresponding convergences of operators stated in the rest of the paper hold along the \emph{same} sequences. Let us remark that for the existence of the limit measure $ \mu $, no assumption on $ \Psi_{\eps} $ is needed, except for its normalization as a vector in $ \Gamma_{\mathrm{s}}(\mathfrak{H})$ (\cref{convcyl}). It is however important to have a probability measure on $\mathfrak{H}$, or on some other related space, as the semiclassical limit of a quantum state. One possible sufficient condition to ensure that any cluster point $\mu$ of $\Psi_{\varepsilon}$ is concentrated as a Radon probability measure is the following: there exists a $\delta \geq 1$ and a finite constant $C = C(\delta) < +\infty $ such that, uniformly in $\eps$,
 \begin{equation}
 	\label{omegacontrol} \tag{A3}
    \meanlr{\Psi_{\eps}}{\mathrm{d}\Gamma(\omega)^{\delta}}{\Psi_{\eps}}_{\Gamma_{\mathrm{s}}(\mathfrak{H})} \leq C(\delta).
\end{equation}
 Obviously, by interpolation, if the above assumption is satisfied for a given $\delta$, then it holds as well for any $ 0 <\underline{\delta} \leq \delta$. In fact, Assumption \ref{omegacontrol} ensures that any cluster point $\mu$ is concentrated on the space
\beq
	\label{eq:homega}
	\mathfrak{H}_{\omega} : = \C^{d-1} \otimes L_{\omega}^2(\mathcal{K}) : = \C^{d-1} \otimes L^2(\mathcal{K}, \omega \diff \nu),
\eeq
(which shares a common dense subset with $\mathfrak{H}$), {\it i.e.}, $ \mu \in \mathscr{M}(\mathfrak{H}_{\omega}) $. In addition,
\beq
\label{eq:en bound}
\int_{\mathfrak{H}_{\omega}} \diff \mu(\zv) \lf\| \sqrt{\omega} \zv \ri\|_{\mathfrak{H}}^{2\underline{\delta}} \leq C(\delta),
\eeq
for any $ 0 \leq \underline{\delta} \leq \delta $ (more details on these infinite dimensional semiclassical techniques, together with bibliographic references, are given in \cref{sec:Wick-quant}).

The partial trace with respect to the degrees of freedom of the field is a standard procedure (see, {\it e.g.}, \cite[Proposition 3.3]{Correggi:2017aa}), which yields, given an operator on $\mathscr{H}=L^2(\Lambda;\mathbb{C}^s)\otimes \Gamma_{\mathrm{s}}(\mathfrak{H})$ and a vector of
$\Gamma_{\mathrm{s}}(\mathfrak{H})$, a quadratic form on $L^2(\Lambda,\mathbb{C}^s)$. Therefore, we set
\begin{equation}
  \label{heps}
  \framebox{$ \mathcal{H}_{\eps} :=\meanlr{\Psi_{\eps}}{H}{\Psi_{\eps}}_{\Gamma_{\mathrm{s}}(\mathfrak{H})} - c_{\eps}\;, $} 
\end{equation}
with
\begin{equation}
	\label{eq:ceps}
  c_{\eps} = \meanlr{\Psi_{\eps}}{\mathrm{d} \Gamma (\omega)}{\Psi_{\eps}}_{\Gamma_{\mathrm{s}}(\mathfrak{H})}.
\end{equation}
Assumption \eqref{omegacontrol} guarantees that $c_{\varepsilon}$ is uniformly bounded with respect to $\varepsilon$, however since it amounts to a simple spectral shift of the effective
Hamiltonian we drop it.

The quasi-classic counterpart of \eqref{heps} is the MSO
\begin{equation}
  \label{eff}
  \framebox{$ \mathcal{H}_{\text{eff}} (\mu) = \disp\sum_{j = 1}^N \lf\{ \bigl(-i\nabla_j - \aav_{j, \mu}(\xv_j)\bigr)^2 - \bm{\sigma}_j \cdot \mathbf{B}_{j,\mu}(\xv_j) + W_{j,\mu}(\xv_j) \ri\} + V(\xv_1, \ldots, \xv_N)\; , $}
\end{equation}
where the classical fields are given by 
\begin{align}
  &\mathbf{A}_{j, \mu}(\xv)= 2 \Re \int_{\mathfrak{H}_{\omega}}^{} \mathrm{d}\mu(\zv)  \braket{\zv}{ \lambda_{j}(\xv) \hat{\eev}}_{\mathfrak{H}} &\text{(``mediated'' magnetic potential)}\; ,	\label{eq:Aclassic}\\
  &\mathbf{B}_{j, \mu}(\xv)= 2 \Re \int_{\mathfrak{H}_{\omega}} \mathrm{d}\mu(\zv)  \braket{\zv}{ b_{j}(\xv) \hat{\eev}}_{\mathfrak{H}} &\text{(``mediated'' magnetic field)}\; \label{eq:Bclassic},
\end{align}
where $ \hat\eev $ stands for the vector (actually a tensor in $ \R^{d-1} \times \R^{d} $)
\beq
	\hat{\eev} = \bigl(\eev_1, \ldots, \eev_{d-1}\bigr),
\eeq
of polarizations. Hence, the scalar product appearing in the above definition is meant, {\it e.g.}, as the vector in $ \R^d $
\bdm
	\braket{\zv}{ \lambda_{j}(\xv) \hat{\eev}}_{\mathfrak{H}} = \sum_{\gamma = 1}^{d-1} \braket{\zv_\gamma}{ \lambda_{j}(\xv) \eev_{\gamma}}_{L^2(\mathcal{K},\diff \nu)} = \sum_{\gamma = 1}^{d-1} \int_{\mathcal{K}} \diff \nu(\kv) \: z_{\gamma}^*(\kv) \lambda_j(\xv;\kv) \eev_{\gamma}(\kv).
\edm
The ``electric" potential $ W_{j,\mu}(\xv) $ is explicitly given by
\beq
	\label{eq:Wclassic}
	W_{j,\mu}(\xv) : = 4 \int_{\mathfrak{H}_{\omega}}^{} \mathrm{d}\mu(\zv)  \lf( \Re \braket{\zv}{ \lambda_{j}(\xv)  \hat{\eev}}_{\mathfrak{H}} \ri)^2 - 4 \lf( \Re  \int_{\mathfrak{H}_{\omega}}^{} \mathrm{d}\mu(\zv) \braket{\zv}{ \lambda_{j}(\xv) \hat{\eev}}_{\mathfrak{H}} \ri)^2
\eeq
and could be thought of as the variance with respect to the measure $ \mu $ of the ``bare'' vector potential
\beq
	\label{eq:fraka}
	\mathfrak{a}_{\zv}(\xv_1, \ldots, \xv_N) := 2\Re \braket{\zv}{\vec{\lambda}}_{\mathfrak{H}}: \Lambda^N \to \mathbb{R}^{dN},
\eeq
i.e., $ \sum_j W_{j,\mu}(\xv_j) = \mu\lf(\mathfrak{a}^2_{\zv}(\xv_1, \ldots, \xv_N) \ri) - \mu\lf(\mathfrak{a}_{\zv}(\xv_1, \ldots, \xv_N)\ri)^2 $.

We can now formulate our first main result: we denote by $ - \Delta_{\mathrm{D}} $ the Dirichlet Laplacian on $ \Lambda^N $, {\it i.e.}, the self-adjoint operator with domain $ \mathscr{D}(-\Delta_{\mathrm{D}}) = H^2_0(\Lambda^N) $, and by $ \mathscr{D}[-\Delta_{\mathrm{D}}+V_+] \subset L^2(\Lambda^N; \C^s) $ the domain of the closed quadratic form associated to the self-adjoint operator $ -\Delta_{\mathrm{D}}+V_+ $; $ \lf\| \: \cdot \: \ri\| - \mathrm{res} $ is short for the convergence in norm resolvent sense.

	\begin{thm}[Effective Hamiltonian]
 		\label{thm:conveff}
 		\mbox{}	\\
  		Let $\bm{\lambda} \in L^{\infty}\bigl(\Lambda^N; \lf(\mathscr{D}[\omega + \omega^{-1}]\ri)^{dN} \bigr)$, $\mathbf{b} \in L^{\infty} \bigl(\Lambda^N; \lf(\mathscr{D}[\omega^{-1}]\ri)^{dN} \bigr)$ and the assumptions \eqref{coulomb}, \eqref{eqn:h3} and \eqref{omegacontrol} be satisfied. In addition, let $ \Psi_{\eps} $  be a normalized vector in $ \Gamma_{\mathrm{s}}(\mathfrak{H}) $, such that $\Psi_{\varepsilon} \to \mu \in \mathscr{M}(\mathfrak{H}_{\omega}) $ in the sense of \eqref{eq:conv}. Then $\mathcal{H}_{\varepsilon}$, $\varepsilon\in (0,1)$, and $\mathcal{H}_{\mathrm{eff}}$ are self-adjoint on the self-adjointness domain of $ -\Delta_{\mathrm{D}}+V_+ $ and
  		\begin{equation}
  			\framebox{$ 
    			\mathcal{H}_{\eps} \xrightarrow[\eps \rightarrow 0]{\| \cdot \| - \mathrm{res}} \mathcal{H}_{\mathrm{eff}} (\mu)\; . $}
  		\end{equation}
	\end{thm}
	
	\begin{remark}[Vector potential]
		\mbox{}	\\
		The vector potential $ \mathbf{A}_{j,\mu} $ depends on the particle index $ j = 1, \ldots, N $, {\it i.e.}, each particle interacts in a different way with the classical field. This is due to the difference in the charge distributions: if all the function $ \lambda_j $, $ j = 1, \ldots, N $, were equal, then the magnetic potential would be independent of $ j $.
	\end{remark}
  
  	\begin{remark}[Electric potential I]
  		\mbox{}	\\
  		The effective Hamiltonian $\mathcal{H}_{\mathrm{eff}}$ is a MSO, that differs however from the naïvely expected form because of the ``variance'' electric potential $ W_{j,\mu} $. The presence of such additional electric-like
potential is motivated by the fact that at the quantum level the expectation $\langle \vec{\varphi}^2 \rangle_{\varepsilon}$ differs in general from
$\langle \vec{\varphi} \rangle_{\varepsilon}^2$. As a matter of fact, we can construct explicit examples
of sequences of quantum Fock vectors for which the above limit is either different from zero (e.g., vectors with an increasing fixed number of
photons, each one in the same one-particle configuration) or equal to zero (e.g., families of squeezed coherent states of minimal
uncertainty). In general, microscopic field configurations whose corresponding classical limit is not a Dirac's delta measure are expected to have a
non-vanishing variance electric potential.  
	\end{remark}
	
	\begin{remark}[Electric potential II]
		\label{rem: positive}
		\mbox{}	\\
		As for the vector potential, the electric potential is in fact particle-dependent, {\it i.e.}, each particle feels a different field. This makes perfect sense in the physical picture, since the particles are assumed to have different charge distributions or form factors. Furthermore, as long as $\mu$ is a probability measure such that $ \mathfrak{a}_{(\: \cdot \:)}(\xv_1, \ldots, \xv_N)$ and
$\mathfrak{a}^2_{(\: \cdot \:)}(\xv_1, \ldots, \xv_N)$ are both measurable and integrable $\xv_1, \ldots, \xv_N \in \Lambda^N$-a.e., the electric potential is pointwise positive a.e.:
	\bml{
            \sum_{j=1}^N W_{j,\mu}(\xv_j) =\mu\lf(\mathfrak{a}^2_{\zv}(\xv_1, \ldots, \xv_N) \ri) - \mu\lf(\mathfrak{a}_{\zv}(\xv_1, \ldots, \xv_N)\ri)^2 \\
            =\mu \lf( \mathfrak{a}^2_{\zv}(\xv_1, \ldots, \xv_N)  - 2 \mu\lf(\mathfrak{a}_{\zv}(\xv_1, \ldots, \xv_N)\ri)  \mathfrak{a}_{\zv}(\xv_1, \ldots, \xv_N) + \mu\lf(\mathfrak{a}_{\zv}(\xv_1, \ldots, \xv_N)\ri)^2 \ri)	\\ = \mu \lf( \lf[\mathfrak{a}_{\zv}(\xv_1, \ldots, \xv_N) -\mu\lf(\mathfrak{a}_{\zv}(\xv_1, \ldots, \xv_N)\ri) \ri]^2 \ri) \geq 0\;.
  		}
  	\end{remark}

	\begin{remark}[Regularity of the potentials]
  		\label{rem:2}
  		\mbox{}	\\
  		With the above assumptions, the effective potential $ \mathbf{A}_{j,\mu}$ and field $\mathbf{B}_{j,\mu}$, as well as the effective electric potential
  $ W_{j,\mu} $ are all regular, {\it i.e.}, continuous and vanishing at infinity. This is apparent also in the fact that the domains of self-adjointness of $ \heps $ and $ \HHe $ coincide. 
\end{remark}

	The last \cref{rem:2} motivates a deeper investigation of the quasi-classical limit: a large class of magnetic fields is indeed excluded from the result in \cref{thm:conveff}, because of the regularity and boundedness of both $ \mathbf{A}_{j,\mu}$ and $\mathbf{B}_{j,\mu}$. For instance, uniform magnetic fields for which $ \aav(\xv) \: \propto \: \xv^{\perp} $ can not be generated in the quasi-classical regime, if the assumptions on the state of the field in \cref{thm:conveff} are met. It is thus intriguing to drop such assumptions and see whether this allows to reach a larger class of MSOs. This is the content of next result, where for the sake of simplicity we drop the spin dependence of the Hamiltonian and assume that the quantum particles are spinless. Let then $ \aav $ be a generic vector potential and $V$ an electric field, such that
\beq
	\label{eq:AV assumptions}
	 \aav \in L^2_{\mathrm{loc}}\lf( \Lambda;\mathbb{R}^{d} \ri),	\qquad		V\in \mathfrak{K}_{<}\lf(\Lambda^N\ri)+L^1_{\mathrm{loc}}\lf(\Lambda^N;\mathbb{R}^+\ri).
\eeq
We denote by $ \mathcal{H}^{\aav,V} $ the corresponding MSO, {\it i.e.},
\begin{equation}
  	\label{eq:generic MSO}
    	\mathcal{H}^{\aav,V} = \sum_{j =1}^N \bigl(-i\nabla_j -\aav(\xv_j)\bigr)^2+V(\xv_1, \ldots, \xv_N),
\end{equation}
which is self-adjoint on a suitable domain contained in $L^2(\Lambda^N)$. We use the short notation $ \mathrm{s-res} $ for the convergence in strong resolvent sense.

	\begin{thm}[Effective Hamiltonian: rougher electromagnetic fields]\label{thm:sconv}
          $\phantom{i}$\\
          Let $ \HH^{\aav,V} $ be any self-adjoint MSO of the form \eqref{eq:generic MSO} with $ \aav $ and $V$ satisfying \eqref{eq:AV
            assumptions}. Then, there exist a self-adjoint microscopic PF-type Hamiltonian $H^{\aav,V}$ on
          $L^2(\Lambda^N)\otimes \Gamma_{\mathrm{s}}(\mathfrak{H})$ of spinless non-point-like quantum charges coupled to quantized electromagnetic
          radiation, and a family $(\Psi_{\varepsilon})_{\varepsilon\in (0,1)}$ of quantum configurations of the radiation field realizing
          $\mathcal{H}^{\aav,V}$ in the quasi-classical limit, {\it i.e.}, denoting by
          $ \mathcal{H}_{\varepsilon}^{\aav,V} : = \meanlrlr{\Psi_{\eps}}{H^{\aav,V} - \diff \Gamma(\omega)}{\Psi_{\eps}}_{\Gamma_{\mathrm{s}}(\mathfrak{H})} $ the
          partial trace of $ H^{\aav,V} - \diff \Gamma(\omega) $,
  		\begin{equation}
    			\framebox{$ \mathcal{H}_{\varepsilon}^{\aav,V} \xrightarrow[\eps \rightarrow 0]{\mathrm{s-res}} \mathcal{H}^{\aav,V} .$}
  		\end{equation}
              \end{thm}

		\begin{remark}[Field states]
			\label{rem:field states}
			\mbox{}	\\
			There are many possible choices of coupling factors $\vec{\lambda}$ and states $ \Psi_{\varepsilon} $ in \cref{thm:sconv} (see Sect. \cref{sec:rough-electr-fields} for some explicit examples). Let us remark here that the sequence of states $(\Psi_{\varepsilon})_{\varepsilon\in (0,1)}$, along which the quasi-classical limit is taken, can be explicitly chosen (see \eqref{eq:field states}), once the magnetic potential $ \aav $ is given, and the coupling factor is fixed. More precisely, we construct $ \Psi_{\eps}$ as a suitable squeezed coherent state, whose argument depends on $ \aav $. A similar strategy was actually followed in \cite{Correggi:2017aa} to show that certain field states can give rise to trapping (electric) potentials in the quasi-classical limit of Nelson-type models. We stress that, exactly as in that case, the energy of the field $ \ceps $ might as well diverge as $ \eps \to 0 $, which is to be expected since the classical field so generated is singular ({\it e.g.}, it is unbounded).
              \end{remark}

        \begin{remark}[Alternative approach]
          \label{rem:1}
          	\mbox{}	\\
           Another way to obtain interesting
          vector potentials, such as the
          aforementioned $\vec{x}^{\perp}$, could
          be to use more singular form
          factors $\vec{\lambda}$. However, there
          are some difficulties connected
          with this strategy, namely that the
          operator $H_{\varepsilon}$ may not be
          self-adjoint or even densely
          defined, and that one should choose
          a very specific combination of
          coupling factor and sequence of
          semiclassical states in order to
          obtain the desired limit. 
        \end{remark}

        We now work out an explicit example in order to clarify the meaning of the above Theorem. The final goal will be the derivation of a magnetic Sch\"{o}dinger operator with uniform magnetic field. Let then 
       	\beq
       		\label{eq:uniform}
        		\aav = \tx\frac{1}{2} \mathbf{x}^{\perp} = \tx\frac{1}{2} (-y, x, 0)
        	\eeq 
        	be the vector potential generating a unitary magnetic field along $ \hat{\mathbf{z}} $ in $ d = 3 $\footnote{The computation can be reproduced also in $ d =2 $ by simply projecting all the quantities on the plane $ x,y $.}, with
          $V=0$, $N=1$. Let also 
          $\{\xi_{\varepsilon}(\mathbf{k})\}_{\varepsilon\in (0,1)}$ be a family of
          compactly supported mollifiers, and set $ \nabla^{\perp} = (-\partial_y,\partial_x,0) $.
          In addition,
          choose
          \begin{equation*}
            \mathbf{\lambda}(\mathbf{x};\mathbf{k})=\sum_{\gamma=1}^{2}\frac{1}{\sqrt{\lvert \mathbf{k}  \rvert_{}^{}}}e^{-\frac{\lvert \mathbf{k}  \rvert_{}^2}{2}}e^{-i \mathbf{k}\cdot \mathbf{x}}\mathbf{e}_{\gamma}(\mathbf{k})\; ;
          \end{equation*}
          and
          \begin{equation*}
            \Psi_{\varepsilon}=\exp\biggl\{\frac{i}{4\varepsilon}\sum_{\gamma=1}^{2}\int_{\R^3} \mathrm{d}\mathbf{k}\; \sqrt{|\kv|}e^{\frac{\lvert \mathbf{k}  \rvert_{}^2}{2}} \lf(a^{\dagger}_{\gamma}(\mathbf{k})\,\lf( \eev_{\gamma} \cdot \nabla^{\perp} \xi_{\varepsilon} \ri) (\mathbf{k}) +a_{\gamma}(\mathbf{k})\lf( \eev_{\gamma} \cdot \nabla^{\perp} \xi_{\varepsilon} \ri)^* (\mathbf{k}) \ri)\biggr\}\,\Omega,
          \end{equation*}
          with $ \Omega \in \Gamma_{\mathrm{s}}(\mathfrak{H}) $ the vacuum vector. Then, \cref{thm:sconv} yields the convergence (in this specific case there is actually no need of sequence extraction)
          \begin{equation*}
            \mathcal{H}_{\varepsilon}^{\aav,0} \xrightarrow[\eps \rightarrow 0]{\mathrm{s-res}} \mathcal{H}^{\aav,0} = \lf(-i\nabla - \tx\frac{1}{2} \xv^{\perp}\ri)^2\;,
          \end{equation*}
          for the vector potential $ \aav $ given by \eqref{eq:uniform}.

It could also be interesting to couple the quasi-classical limit on the field with a mean field limit on the particles, \emph{i.e.}, to let $N\to \infty$. In that case, either if we take the limit $N\to \infty$ before or after the limit $\varepsilon\to 0$, the Hamiltonian converges to the same non-linear effective energy functional, describing the mean-field interaction of one particle with the classical field. The effective model is in this case the same obtained considering the coupled limit $N \propto \frac{1}{\varepsilon} $, as $ \eps \to 0 $. Such a coupled limit has been studied in the dynamical setting for the Nelson and Pauli-Fierz models in \cite{falconi2013jmp,Ammari:2014aa,ammari:15}  and  \cite{leopold2016arxiv}, respectively.
        
\subsection{Ground state energy}

In this section we study the behaviour of the ground state energy of the microscopic Hamiltonian \eqref{eqn:ham} in the quasi-classical
limit $\eps \rightarrow 0$. The ground state energy of Pauli-Fierz-type operators is defined as
\beq
	\underline{\sigma}(H_{\eps}) := \inf \bigl\{ \lambda \in \RR \; | \; \lambda \in \sigma(H_{\varepsilon})\bigr\}
\eeq
where we have emphasized the dependence on $ \eps $ of $ H $ by adding a label. Concretely, the quantity can be computed via a suitable minimization of the energy quadratic form, {\it e.g.},
\[
	\underline{\sigma}(H_{\eps}) = \inf_{\substack{\psi \in \mathscr{D}[H_{\eps}], \: \|\psi \|_{\mathscr{H}} = 1}} \langle \psi | H_{\eps} | \psi \rangle \; 
\]
with $ \mathscr{D}[H_{\eps}] $ the form domain of $ H_{\eps} $ (or any core for it).

Analogously, it is possible to define the ground state energy for the effective models
$\underline{\sigma}\bigl(\mathcal{H}_{\mathrm{eff}}(\mu)\bigr)$, taking into account that the effective Hamiltonians depend on the Wigner measure
$\mu\in \mathfrak{H}_{\omega} $, describing the classical state of the field. Heuristically, the ground state energy of the microscopic system is expected to converge in the quasi-classical
limit to the infimum over all possible classical configurations, which can be obtained as classical limits of vectors in the domain of $H_{\varepsilon}$, of the
effective ground state energies. In fact, as we are going to see, it is sufficient to consider the smaller minimization domain (recall the definition of $ \mathfrak{H}_{\omega} $ in \eqref{eq:homega})
\begin{equation}\label{measures}
\mathscr{M}_{2,\omega}:= \mathscr{M}_{2} \lf(\mathfrak{H}_{\omega}\ri)= \lf\{ \mu \in \mathscr{M}\lf(\mathfrak{H}_{\omega}\ri) \; \bigg | \;  \int_{\mathfrak{H}_{\omega}} \diff \mu(\zv) \lf\| \sqrt{\omega} \zv \ri\|_{\mathfrak{H}}^{2} < + \infty \ri\}.
\end{equation}
The second moment of $\mu$, which is finite for any $\mu\in \mathscr{M}_{2,\omega}$, is also the classical limit of the average free energy of photons (see \cref{lem:ceps}):
\begin{equation}
	\label{eq:cmu}
  	c(\mu) :=\int_{\mathfrak{H}_{\omega}}^{}  \diff \mu(\zv) \: \lf\| \sqrt{\omega} \zv \ri\|_{\mathfrak{H}}^{2} =\lim_{\varepsilon\to 0}c_{\varepsilon}\; ,
\end{equation}
provided $\Psi_{\varepsilon}\to \mu$ in the sense of \eqref{eq:conv}. We recall that $c_{\varepsilon}$ given by \eqref{eq:ceps} is the energy of the field on the state $ \Psi_{\eps} $.

\begin{thm}[Ground state energy convergence]
	\label{thm:2}
	\mbox{}	\\
  	Let $\bm{\lambda} \in L^{\infty}\bigl(\Lambda^N; \lf(\mathscr{D}[\omega + \omega^{-1}]\ri)^{dN} \bigr)$, $\mathbf{b} \in L^{\infty} \bigl(\Lambda^N; \lf(\mathscr{D}[\omega^{-1}]\ri)^{dN} \bigr)$ and the assumptions \eqref{coulomb} and \eqref{eqn:h3} be satisfied. Then,
	\begin{equation}
  		\label{bottom}
		\framebox{$ \disp\lim_{\eps \rightarrow 0} \underline{\sigma}(H_{\varepsilon}) = \inf_{\mu \in \mathscr{M}_{2,\omega}} \Bigl(\underline{\sigma}\lf(\HHe(\mu)\ri) + c(\mu)\Bigr)\; . $}
	\end{equation}
\end{thm}

\begin{remark}[Boundedness from below]
	\label{rem:bd below}
	\mbox{}	\\
	Since by \cref{selfadthm} $ H_{\eps} $ is bounded from below with bound uniform in $ \eps $, a byproduct of \cref{thm:2} is that the r.h.s. of \eqref{bottom} is also a quantity which is bounded from below. This was indeed not a priori obvious and, in order to hold true, the presence of the field energy $ c(\mu) $ is in fact crucial.
\end{remark}

\section{Convergence in the Quasi-classical Limit}
\label{sec:convergence}
   
This Section is devoted to the proof of \cref{thm:conveff} and \ref{thm:sconv} on the convergence of the partial trace of the PF Hamiltonian to a suitable effective MSO.

\subsection{Convergence to the effective Hamiltonian}

Proof of \cref{thm:conveff} is split in several steps: we first identify the operator $ \heps $ given by \eqref{heps}; next we show that it is self-adjoint on a suitable domain, where the effective operator $ \HHe $ is self-adjoint too; then we prove pointwise convergence as $ \eps \to 0 $ and the proof is thus completed by exploiting a dominated convergence argument.

We first notice that, under the assumptions  \eqref{coulomb}, \eqref{eqn:h3} and \eqref{omegacontrol}, we can explicitly identify the partial trace of $ H $, {\it i.e.}, the operator \eqref{heps}, which takes the form
\bml{
  	\label{hepsexplicit}
  	\mathcal{H}_{\eps} = -\Delta_{\mathrm{D}} + \lf\langle\bm{\varphi}^2\ri\rangle_{\eps}(\xv_1, \ldots, \xv_N) + i \bigl(\mathbf{\alpha}^{*}_{\varepsilon}(\xv_1, \ldots, \xv_N) \cdot \nabla+\nabla\cdot \vec{\alpha}_{\varepsilon}(\xv_1, \ldots, \xv_N)\bigr)\\
   -\vec{\sigma}\cdot \mathbf{B}_{\varepsilon}\lf(\xv_1, \ldots, \xv_N\ri) +V(\xv_1, \ldots, \xv_N)\; 
}
where 
\begin{gather}
  \lf\langle\bm{\varphi}^2\ri\rangle_{\varepsilon}\lf(\xv_1, \ldots, \xv_N\ri)= \meanlr{\Psi_{\varepsilon}}{\bm{\varphi}(\vec{\lambda}(\xv_1, \ldots, \xv_N)) \cdot \bm{\varphi}(\vec{\lambda}(\xv_1, \ldots, \xv_N))}{\Psi_{\varepsilon}}_{\Gamma_{\mathrm{s}}(\mathfrak{H})} ,\\
  \vec{\alpha}_{\varepsilon}\lf(\xv_1, \ldots, \xv_N\ri) : = \braketl{\Psi_{\varepsilon}}{ a(\vec{\lambda}\lf(\xv_1, \ldots, \xv_N\ri))  \Psi_{\varepsilon}}_{\Gamma_{\mathrm{s}}(\mathfrak{H})} ,	\\
  	\mathbf{B}_{\varepsilon}\lf(\xv_1, \ldots, \xv_N\ri) : = \meanlr{\Psi_{\varepsilon}}{\bm{\varphi}\lf(\mathbf{b}\lf(\xv_1, \ldots, \xv_N\ri)\ri)}{ \Psi_{\varepsilon}}_{\Gamma_{\mathrm{s}}(\mathfrak{H})} .
\end{gather}
We also denote by $ \aav_{j,\eps} $, $ \mathbf{B}_{j,\eps} $ the components of $ \vec{\alpha}_{\varepsilon} + \vec{\alpha}_{\eps}^*$, $ \mathbf{B}_{\varepsilon} $ respectively, {\it i.e.},
\begin{gather}
  	\aav_{j,\eps}(\xv) : = \sum_{\gamma = 1}^{d-1} \meanlrlr{\Psi_{\eps}}{a^{\dagger}_{\gamma} \lf(\lambda_{j,\gamma}(\xv) \mathbf{e}_{\gamma} \ri) +a_{\gamma} \lf(\lambda_{j,\gamma}(\xv) \mathbf{e}_{\gamma} \ri)}{\Psi_{\eps}}_{\Gamma_{\mathrm{s}}(\mathfrak{H})} , \\
  		\mathbf{B}_{j,\eps}(\xv) : = \sum_{\gamma = 1}^{d-1} \meanlrlr{\Psi_{\eps}}{a^{\dagger}_{\gamma} \lf(b_{j,\gamma}(\xv) \mathbf{e}_{\gamma} \ri) +a_{\gamma} \lf(b_{j,\gamma}(\xv) \mathbf{e}_{\gamma} \ri)}{\Psi_{\eps}}_{\Gamma_{\mathrm{s}}(\mathfrak{H})}.
\end{gather}

We start by proving self-adjointness of $\mathcal{H}_{\varepsilon}$, $\varepsilon\in (0,1)$, and $\mathcal{H}_{\mathrm{eff}}$ with the aforementioned assumptions
(for a matter of convenience, we omit the explicit dependence on $\mu$ of $\mathcal{H}_{\mathrm{eff}}(\mu)$ in the rest of this section). We denote by $ q_{\heps} $ and $ q_{\HHe} $ the quadratic forms associated to $ \heps $ and $ \HHe $ respectively.

	\begin{proposition}[Self-adjointness of $ \heps $ and $ \HHe $]
 		\label{prop:1}
  		\mbox{}	\\
 		Let $q_{0,\mathrm{D}}[\; \cdot \;] =\meanlr{\, \cdot \,}{- \Delta_{\mathrm{D}} +V_+}{\,\cdot \,}_{L^2(\Lambda^N;\mathbb{C}^s)}$  be the closed form with domain
  $ \mathscr{D}[-\Delta_{\mathrm{D}}+V_+]$. Then, the quadratic forms $ q_{\mathcal{H}_{\varepsilon}}-q_{0,\mathrm{D}} $, $\varepsilon\in (0,1)$, and
  $ q_{\mathcal{H}_{\mathrm{eff}}}-q_{0,\mathrm{D}} $ are Kato-small w.r.t. $q_{0,\mathrm{D}}$.
	\end{proposition}
	
	\begin{proof}
  		By hypothesis, we already know that $V_{<}\in \mathfrak{K}_{<}(\Lambda^N)$ is a Kato-small perturbation. We now show that, for any $ j = 1, \ldots, N $,
  		\begin{equation}
    		\label{eq:1}
    			\begin{gathered}
      			\vec{\alpha}_{\varepsilon}\in L^{\infty}\lf(\Lambda^N;\mathbb{C}^{dN}\ri);	\qquad	\mathbf{B}_{j,\varepsilon}\,,\,\aav_{j,\mu}\,,\,\mathbf{B}_{j,\mu}\in L^{\infty}\lf(\Lambda;\mathbb{R}^{dN}\ri);\\ \lf\langle\bm{\varphi}^2\ri\rangle_{\eps},\; W_{j,\mu} \in L^{\infty}\lf(\Lambda_N,\mathbb{R}^+\ri),
    			\end{gathered}
  		\end{equation}
  		uniformly w.r.t. $\varepsilon\in (0,1)$. Since the proofs are basically the same, for the sake of brevity, we take into account only one $\varepsilon$-function and one $\mu$-function. Let us consider first $ \lf\langle\bm{\varphi}^2\ri\rangle_{\eps} $:
  		\bmln{
      		\lf| \lf\langle\bm{\varphi}^2\ri\rangle_{\eps}(\xv_1, \ldots, \xv_N)  \ri| \leq \sum_{j=1}^{N} \sum_{\gamma =1}^{d-1} \bigl\| \varphi(\lambda_j(\xv_j) \eev_{\gamma}) \Psi_{\varepsilon} \bigr\|_{\Gamma_{\mathrm{s}}(\mathfrak{H})}^2	\\
      		\leq 2\sum_{j=1}^{N} \sum_{\gamma =1}^{d-1} \lf[ \bigl\| a(\lambda_j(\xv_j) \eev_{\gamma}) \Psi_{\varepsilon} \bigr\|_{\Gamma_{\mathrm{s}}(\mathfrak{H})}^2+\bigl\| a^{\dagger}(\lambda_j(\xv_j) \eev_{\gamma}) \Psi_{\varepsilon} \bigr\|_{\Gamma_{\mathrm{s}}(\mathfrak{H})}^2 \ri] \\ 
      		\leq 4 d \sum_{j=1}^{N} \sup_{\xv \in \R^d} \lf[ \bigl\| \omega^{-1/2} \lambda_j(\xv) \bigr\|_{L^2(\mathcal K,\diff \nu)}^2 \lf\| \mathrm{d}\Gamma(\omega)^{1/2} \Psi_{\varepsilon}  \ri\|_{\Gamma_{\mathrm{s}}(\mathfrak{H})}^2 + \eps \lf\| \lambda_{j}(\xv) \ri\|_{L^2(\mathcal K, \diff \nu)}^2 \ri].
    		}
  		It then follows that
  		\begin{equation}
    			\lf\| \lf\langle\bm{\varphi}^2\ri\rangle_{\eps}  \ri\|_{L^{\infty}\lf( \Lambda^N\ri)} \leq 4  \lf[ C(1) \bigl\| \omega^{-1/2} \vec{\lambda}( \: \cdot \:) \bigr\|^2_{L^{\infty}\lf( \Lambda^N; \mathfrak{H} \ri)} + \eps \bigl\| \vec{\lambda}( \: \cdot \:) \bigr\|_{L^{\infty}\lf( \Lambda^N; \mathfrak{H} \ri)}^2 \ri]\; ,
  		\end{equation}
  where $C(1)$ is the constant $C(\delta)$ defined in \eqref{omegacontrol} for $ \delta = 1 $. The estimate of $\aav_{j,\mu}$ is very similar:
 	 \begin{equation}
    		\begin{split}
      		\lf| \aav_{j,\mu}(\xv) \ri|_{}^{} \leq 2 \int_{\mathfrak{H}_{\omega}}^{} \mathrm{d}\mu(\zv)  \lf| \braket{\zv}{ \lambda_{j}(\xv)}_{\mathfrak{H}} \ri| \leq 2\bigl\| \omega^{-1/2} \vec{\lambda}( \: \cdot \:) \bigr\|_{L^{\infty}\lf( \Lambda^N; \mathfrak{H} \ri)} \int_{\mathfrak{H}_{\omega}}^{}  \diff \mu(\zv) \: \lf\| \sqrt{\omega} \zv \ri\|_{\mathfrak{H}} \\\leq 2 C(\tx\frac{1}{2}) \bigl\| \omega^{-1/2} \vec{\lambda}( \: \cdot \:) \: \bigr\|^2_{L^{\infty}\lf( \Lambda^N; \mathfrak{H} \ri)},
    		\end{split}
  	\end{equation}
  	again by \eqref{omegacontrol} (see also \eqref{eq:en bound}). Using then the assumption \eqref{eq:lam} on $ \vec{\lambda} $, the result easily follows. In the case of quantities depending on $ \mathbf{b} $, one can easily reproduce the argument and use \eqref{eq:b1}.
  	
  	Once the uniform boundedness of the quantities in \eqref{eq:1}, it is straightforward to get the final result, {\it i.e.}, Kato-smallness of the quadratic forms in the statement (in fact, infinitesimal Kato-smallness of all the terms but $ V_< $).
\end{proof}

Next, we prove the pointwise convergence a.e. w.r.t. $ \xv \in \Lambda $ of $\aav_{j,\varepsilon} $ and
$\mathbf{B}_{j,\varepsilon}$ to $\aav_{j,\mu}$ and $\mathbf{B}_{j,\mu}$ respectively. Similarly, we also prove the a.e. pointwise convergence in $ \Lambda^N $ of $\langle\vec\varphi^2\rangle_{\varepsilon} $ to $ \mu\lf(\mathfrak{a}^2_{\zv} \ri) $ (see \eqref{eq:fraka} for the definition).

	\begin{lemma}[Pointwise convergence]
  		\label{point}
  		\mbox{}	\\
  		Let $\Psi_{\eps} \in \Gamma_{\mathrm{s}}(\mathfrak{H})$ normalized satisfy \eqref{omegacontrol}, so that $\Psi_{\varepsilon}\to \mu\in \mathscr{M}\bigl(\mathfrak{H}_{\omega}\bigr)$ in the sense of \eqref{eq:conv}. Then, for a.e. $ \xv \in \Lambda$ and for any $ j = 1, \ldots, N $,
  		\beq
  			\aav_{j,\eps}(\xv) \xrightarrow{ \eps \rightarrow 0} \aav_{j,\mu} (\xv),	\qquad \mathbf{B}_{j,\eps}(\xv) \xrightarrow{ \eps \rightarrow 0} \mathbf{B}_{j,\mu} (\xv).
		\eeq
		Moreover, for a.e. $ \xv_1, \ldots, \xv_N \in \Lambda^N $,
   		\beq
   			\lf\langle\vec\varphi^2\ri\rangle_{\varepsilon} (\xv_1, \ldots, \xv_N) \xrightarrow[\eps \rightarrow 0]{}  \mu\lf(\mathfrak{a}^2_{\zv}\lf(\xv_1, \ldots, \xv_N \ri)\ri).
		\eeq
  	\end{lemma}

	\begin{proof}
          The idea is to use the tools of semiclassical analysis in infinite dimensions outlined in \cref{sec:Wick-quant} (see
          \cite{ammari:nier:2008,Falconi:2017aa} for additional details). Let us start with $\aav_{\varepsilon}$: we need to prove that the
          corresponding Fock space operator-valued vector $\vec\varphi(\vec{\lambda})$ is actually the Wick-quantization of a compact (finite-rank)
          symbol.
  
  		Using the correspondences (\ref{eqn:a}) and (\ref{eqn:a*}), we may write for any $ j = 1, \ldots, N $, $ \gamma = 1, \ldots, d-1 $ and a.e. $\xv \in \Lambda$:
  		\bmln{
   			\vec\varphi_{j}(\lambda_j(\xv)) = \sum_{\gamma = 1}^{d-1} \lf[ \braket{\lambda_j(\xv) \eev_{\gamma}}{z_{\gamma}}_{L^2(\mathcal{K},\diff \nu)}^{\text{Wick}} + \braket{z_{\gamma}}{\lambda_j(\xv) \eev_{\gamma}}_{L^2(\mathcal{K},\diff \nu)}^{\text{Wick}} \ri] \\
   			= \braket{\lambda_j(\xv) \hat{\eev}}{\zv}_{\homega}^{\text{Wick}} + \braket{\zv}{\lambda_j(\xv) \hat{\eev}}_{\homega}^{\text{Wick}}= : \sv_{1,\xv}(\zv)^{\text{Wick}}+\sv_{2,\xv}(\zv)^{\text{Wick}}\;,
		}
  		where we have kept track of the dependence on the particle coordinate $ \xv \in \Lambda $. Notice that we have used the compact notation $ \zv = \big( z_1, \ldots, z_{d-1} \big) \in \hh $. The classical functions $ \sv_{i,\xxv} $, $ i = 1,2 $, are actually polynomial symbols in the sense of \cref{def:poly}: one can easily see that it is possible to write
		\bdm
			\sv_{1,\xv}(\zv) = \braket{1}{\tilde{\sv}_{1,\xv} \: \zv}_{\mathbb{C}},	\qquad		\sv_{2,\xv}(\zv) = \braket{\zv}{\tilde{\sv}_{2,\xv} \: 1}_{\mathfrak{H}},
		\edm
		where $ \tilde{\sv}_{1,\xv}: \mathfrak{H} \cup \mathfrak{H}_{\omega} \to \C $ is simply the projection onto $ \lambda_j(\xv) \eev_{\gamma} $ component-wise, while $ \tilde{\sv}_{2,\xv}: \C \to \mathfrak{H} \cup \mathfrak{H}_{\omega} $ multiplies any complex number by $  \lambda_j(\xv) \eev_{\gamma}  $ again component-wise. Under the assumptions we made on $ \vec\lambda $, both operators $ \tilde{\sv}_{i,\xv}  $, $ i = 1,2 $, are bounded and, in fact, $ \tilde{\sv}_{2,\xv} = \tilde{\sv}_{1,\xv}^{\dagger} $. Furthermore, both $\tilde{\sv}_{1,\xv} $ and $ \tilde{\sv}_{2,\xv} $ are finite-rank and thus compact. For the first one it is obvious, but also the range of $ \tilde{\sv}_{2,\xv} $ is the one-dimensional subspace spanned by $  \lambda_j(\xv) \eev_{\gamma} $, $ \gamma = 1, \ldots, d-1 $. In conclusion, for any $ \xv \in \Lambda $, $ \sv_{1,\xv} \in \mathscr{P}^{\infty}_{1,0}\bigl(\mathfrak{H}_{\omega} \bigr)$ and
  $\sv_{2,\xv}\in \mathscr{P}^{\infty}_{0,1}\bigl(\mathfrak{H}_{\omega} \bigr)$.
  
    		Applying now \cref{eqn:comp}, we obtain
 		\bml{
    			\aav_{j,\varepsilon}(\xv) = \meanlrlr{\Psi_{\eps}}{\vec\varphi_{j}\lf(\lambda_j(\xv)\ri)}{\Psi_{\eps}}_{\fock} \\
    			= \meanlrlr{\Psi_{\eps}}{\sv_{1,\xv}(\zv)^{\text{Wick}}+\sv_{2,\xv}(\zv)^{\text{Wick}}}{\Psi_{\eps}}_{\fock} \xrightarrow[\eps \rightarrow 0]{} \int_{\mathfrak{H}_{\omega}}^{} \mathrm{d}\mu(\zv) \: \lf( \sv_{1,\xv}(\zv) +\sv_{2,\xv}(\zv) \ri) \\
    			= 2 \Re \int_{\mathfrak{H}_{\omega}}^{} \mathrm{d}\mu(\zv)  \braket{\zv}{ \lambda_{j}(\xv) \hat{\eev}}_{\mathfrak{H}} = \aav_{j,\mu}(\xv)\;,
 		}
  		for a.e. $  \xv \in \Lambda $. 
  		
  		The proof for $\mathbf{B}_{\varepsilon}$ is perfectly analogous. Let us outline the proof for $\langle \vec\varphi^2\rangle_{\varepsilon}$. Using the canonical commutation relations and
  \eqref{corresp}, we obtain
  		\begin{align*}
  			\vec{\varphi}(\vec{\lambda}) \cdot \vec{\varphi}(\vec{\lambda}) 	&= \sum_{j = 1}^N \lf[ \lf( a(\lambda_j(\xv_j) \hat{\eev}) \ri)^2 + \lf( a^{\dagger}(\lambda_j(\xv_j) \hat{\eev}) \ri)^2 +2a^{\dagger}(\lambda_j(\xv_j) \hat{\eev})\cdot a(\lambda_j(\xv_j) \hat{\eev}) + \varepsilon \lf\|  \lambda_j \hat{\eev} \ri\|_{\mathfrak{H}^{d}}^2 \ri]	\\
		    & = : \sum_{j = 1}^N \lf[ s_{3,\xv_j}(\zv)^{\text{Wick}}+ s_{4,\xv_j}(\zv)^{\text{Wick}}+ s_{5,\xv_j}(\zv)^{\text{Wick}}+ \varepsilon \lf\|  \lambda_j \hat{\eev} \ri\|_{\mathfrak{H}^{d}}^2 \ri],
  		\end{align*} 
  		where $ a(\mathbf{v})^2 = a(\mathbf{v}) \cdot a(\mathbf{v}) $ and similarly for $ a^{\dagger} $, and we have omitted the dependence on $ j $ of the symbols to simplify the notation, since the behavior is perfectly the same. The above symbols have the following explicit forms
  		\begin{equation*}
    			s_{3,\xv}(\zv) = \braketr{1}{\tilde{s}_{3,\xv} \zv \otimes \zv}_{\mathbb{C}}\; ,\; s_{4,\xv}(\zv)= \braketr{\zv \otimes \zv}{\tilde{s}_{4,\xv}\, 1}_{\mathfrak{H}\otimes _{\mathrm{s}}\mathfrak{H}}\; ,\; s_{5,\xv}(\zv)= \braketr{\zv}{\tilde{s}_{5,\xv} \, \zv}_{\mathfrak{H}}\; ;
 		\end{equation*}
  		where $ \tilde{s}_{3,\xv}: \lf( \mathfrak{H}  \otimes_{\mathrm{s}} \mathfrak{H} \cup \homega   \otimes_{\mathrm{s}}  \homega \ri)^2 \to \C  $ and $ \tilde{s}_{4,\xv}: \C \to \lf( \mathfrak{H}  \otimes_{\mathrm{s}} \mathfrak{H} \cup \homega   \otimes_{\mathrm{s}}  \homega \ri)^2 $ are actually one the adjoint of the other, {\it i.e.}, $ \tilde{s}_{4,\xv} = \tilde{s}_{3,\xv}^{\dagger} $, and $ \tilde{s}_{5,\xv} $ maps $ \mathfrak{H}  \otimes_{\mathrm{s}} \mathfrak{H} \cup \homega   \otimes_{\mathrm{s}}  \homega $ to itself.  Following the same reasoning as above, one can easily show that any $\tilde{s}_{j,\xv}$, $j=3,4,5$, is finite rank. Let us prove it explicitly for $\tilde{s}_{5,\xv}$: it
  has the form of a (non-orthogonal) projector
  		\begin{equation*}
    			\ket{\lambda_{j}(\xv)}\bra{\lambda_{j}(\xv)},
  		\end{equation*}
  		and thus its range is the one-dimensional space $ \mathrm{span}_{\mathbb{C}}\{ \lambda_{x,j} \} $. Therefore, we conclude that
  $ s_{3,\xv}(\zv)\in \mathscr{P}^{\infty}_{2,0}\bigl(\homega\bigr)$,
  $ s_{4,\xv}(\zv) \in \mathscr{P}^{\infty}_{0,2}\bigl(\homega\bigr)$, and
  $ s_{5,\xv}(\zv) \in \mathscr{P}^{\infty}_{1,1}\bigl(\homega \bigr)$ for any $ j =1, \ldots, N $ and a.e. $ \xv \in \Lambda $. Then again by \cref{eqn:comp}, we get the convergence
  		\begin{equation*}
    			\meanlrlr{\Psi_{\eps}}{\vec{\varphi}(\vec{\lambda}) \cdot \vec{\varphi}(\vec{\lambda})}{\Psi_{\eps}}_{\fock} \xrightarrow[\eps \rightarrow 0]{} \mu\lf(\mathfrak{a}^2_{\zv}\lf(\xv_1, \ldots, \xv_N \ri)\ri).
  		\end{equation*}
	\end{proof}

In the proof of \cref{thm:conveff}, we also need the following technical estimate.

	\begin{lemma}
  		\label{lemma:1}
  		\mbox{}	\\
  		For any $0\leq r\leq 1$, $ -\xi \in \varrho(\mathcal{H}_{\mathrm{eff}})$, and $\phi\in L^2 (\Lambda^N;\mathbb{C}^s)$,
  $(\mathcal{H}_{\mathrm{eff}}+\xi)^{-1}\phi\in H^r(\Lambda^N;\mathbb{C}^s)$. In particular, there exists a finite constant $C$, such that
  		\begin{equation*}
    			\bigl\lVert (\mathcal{H}_{\mathrm{eff}}+\xi)^{-1}\phi  \bigr\rVert_{H^r(\Lambda^N;\mathbb{C}^s)}^{}\leq C\bigl\lVert \phi  \bigr\rVert_{L^2(\Lambda^N;\mathbb{C}^s)}^{}\; .
 		\end{equation*}
	\end{lemma}

	\begin{proof}
  		For $r=0$, the result holds trivially. It is then sufficient to prove it for $r=1$, and obtain all the intermediate cases by interpolation, or,
  since we do not need to optimize the constants, simply using that $(1+\lvert \kkv \rvert_{}^2)^{r}\leq (1+\lvert \kkv \rvert_{}^2)^{r'}$, for any
  $r\leq r'$ ($\kkv$ being the dual variable of $ \xxv \in \Lambda^N$).

  		For $r=1$, we can use the trivial identity for any $ \psi \in H^2(\Lambda^N;\C^s) $
  		\begin{equation*}
    			\lVert \psi  \rVert_{H^1(\Lambda^N;\mathbb{C}^s)}^2= \meanlrlr{\psi}{ 1-\Delta_{\mathrm{D}}}{\psi}_{L^2(\Lambda^N;\mathbb{C}^s)}\;
  		\end{equation*}
 		where we recall that $ \Delta_{\mathrm{D}} := \Delta_{\mathrm{D}} \otimes \mathds{1}_{s\times s}$ is the Laplacian with domain $ H^2_0(\Lambda^N; \C^s) $. In fact, by positivity of the Laplacian, the above identity also holds true as an inequality in the whole $ L^2(\Lambda^N;\C^s) $. Therefore, assuming without loss of generality that $ \xi \in \R $, we have that
  		\bmln{
      		\bigl\lVert (\mathcal{H}_{\mathrm{eff}}+\xi)^{-1}\phi  \bigr\rVert_{H^1(\Lambda^N;\mathbb{C}^s)}^2 \leq  \meanlr{(\mathcal{H}_{\mathrm{eff}}+\xi)^{-1}\phi}{1-\Delta_{\mathrm{D}}+ V_+}{(\mathcal{H}_{\mathrm{eff}}+\xi)^{-1}\phi}_2 \\
      = \braketr{(\mathcal{H}_{\mathrm{eff}}+\xi)^{-1}\phi}{(\HHe + \xi)(\mathcal{H}_{\mathrm{eff}}+\xi)^{-1}\phi}_2	\\
      +  \sum_{j = 1}^N \meanlr{(\mathcal{H}_{\mathrm{eff}}+\xi)^{-1}\phi}{-2i \aav_{j,\mu}(\xv_j)\cdot \nabla_j + \vec{\sigma}_j \cdot \mathbf{B}_{j,\mu}(\xv_j)}{(\mathcal{H}_{\mathrm{eff}}+\xi)^{-1}\phi}_2	\\
      -  \meanlr{(\mathcal{H}_{\mathrm{eff}}+\xi)^{-1}\phi}{\mu\lf(\mathfrak{a}^2_{\zv}\lf(\xv_1, \ldots, \xv_N \ri)\ri) + V_{<}(\xv_1,\ldots, \xv_N)}{(\mathcal{H}_{\mathrm{eff}}+\xi)^{-1}\phi}_2	\\
       + (1-\xi) \lVert (\mathcal{H}_{\mathrm{eff}}+\xi)^{-1}\phi  \rVert_2^2\; .
		}
 		Hence, by the boundedness of $ (\HHe + \xi)^{-1} $  for any $ -\xi \in \varrho(\HHe) $, 
  		\bmln{
      			\bigl\lVert (\mathcal{H}_{\mathrm{eff}}+\xi)^{-1}\phi  \bigr\rVert_{H^1(\Lambda^N;\mathbb{C}^s)}^2 \leq \lf(\delta + \alpha_{<} \ri) \bigl\lVert (\mathcal{H}_{\mathrm{eff}}+\xi)^{-1}\phi  \bigr\rVert_{H^1(\Lambda^N;\mathbb{C}^s)}^2 \\
      			+ C \bigg[ 1 + \sum_{j=1}^N \lf( \tx\frac{1}{\delta} \lf\| \aav_{j,\mu} \ri\|_{L^{\infty}(\Lambda;\mathbb{R}^{d})}^2 + \lf\| \vec\sigma_j \ri\|_{M_{s\times s}(\C)} \lf\| \mathbf{B}_{j,\mu} \ri\|_{L^{\infty}(\Lambda;\mathbb{R}^{d})} \ri) + \lf\| \mu\lf(\mathfrak{a}^2_{\zv}\ri) \ri\|_{L^{\infty}(\Lambda^N)} \bigg] \lVert \phi  \rVert_2^2 	\\
      			\leq \lf(\delta + \alpha_{<} \ri) \bigl\lVert (\mathcal{H}_{\mathrm{eff}}+\xi)^{-1}\phi  \bigr\rVert_{H^1(\Lambda^N;\mathbb{C}^s)}^2 
      			+ C \lVert \phi  \rVert_2^2,
    		}
  		where $\delta >0$ is arbitrary, $0 < \alpha_{<}<1$ is the relative form bound of   $V_{<}$ w.r.t. $ -\Delta_{\mathrm{D}} $, and we have included the estimate of the norm of $ (\HHe + \xi)^{-1} $ in the constant $ C $. We now pick $\delta$ small enough, in such a way that $ \delta +\alpha_{<} < 1$. Then, we get
  		\bdm
  		      \bigl\lVert (\mathcal{H}_{\mathrm{eff}}+\xi)^{-1}\phi  \bigr\rVert_{H^1}^2\leq C \lf(1-\delta - \alpha_{<} \ri)^{-1}\lVert \phi  \rVert_2^2,
   		\edm
  		which yields the result.
	\end{proof}

We are now in a position to prove \cref{thm:conveff}.

	\begin{proof}[Proof of \cref{thm:conveff}]
  		Let $- \xi \in \varrho(\mathcal{H}_{\eps}) \cap \varrho(\mathcal{H}_{\text{eff}})$ uniformly in $\eps$, \emph{i.e.}, there exists $C>0$ so that
  $\text{dist}(-\xi, \sigma(\mathcal{H}_{\eps})) > C $. By the second resolvent identity,
		\bmln{
			I_{\varepsilon}:=\lf\| (\mathcal{H}_{\eps} + \xi)^{-1} - (\mathcal{H}_{\text{eff}} + \xi)^{-1}  \ri\| \\
			= \sup_{\substack{\psi,\phi \,\in L^2(\Lambda^N;\mathbb{C}^s)\\ \lf\|\psi \ri\|_2, \lf\|\phi \ri\|_2 \leq 1}} \lf| \meanlrlr{\psi}{\lf(\mathcal{H}_{\eps} + \xi \ri)^{-1} \lf(\mathcal{H}_{\eps}-\mathcal{H}_{\text{eff}} \ri) \lf(\mathcal{H}_{\text{eff}} + \xi \ri)^{-1}}{\phi} \ri|,
		}
		and therefore
		\bmln{
    			I_{\varepsilon} \leq \sup_{\substack{\psi,\phi \,\in L^2(\Lambda^N;\mathbb{C}^s)\\ \lf\|\psi \ri\|_2, \lf\|\phi \ri\|_2 \leq 1}} \bigg\{ \lf| \meanlrlr{(\mathcal{H}_{\eps} + \xi)^{-1} \psi}{ \lf\langle\vec\varphi^2\ri\rangle_{\varepsilon}  -\mu\lf(\mathfrak{a}^2_{\zv}\ri)}{(\mathcal{H}_{\text{eff}} + \xi)^{-1} \phi} \ri| \\
    			+\sum_{j = 1}^N \lf| \meanlr{(\mathcal{H}_{\eps} + \xi)^{-1} \psi}{i \nabla_j \cdot \lf(\aav_{j,\eps} -\aav_{j,\mu} \ri)  -\vec{\sigma}_j \cdot \lf(\mathbf{B}_{j,\eps}  - \mathbf{B}_{j,\mu}  \ri)}{(\mathcal{H}_{\text{eff}} + \xi)^{-1} \phi} \ri| \bigg\}.
  		}

		The idea is now to exploit the $\xv$-a.e. pointwise convergence proven in \cref{point} to show the convergence of the expression above. The strategy is in fact common to all the terms, so let us consider only a single summand of the second term on the r.h.s., {\it i.e.}, explicitly 
		\bdm
  			I_{\varepsilon}(\aav):=\sup_{\substack{\psi,\phi \,\in L^2(\Lambda^N;\mathbb{C}^s)\\ \lf\|\psi \ri\|_2, \lf\|\phi \ri\|_2 \leq 1}} \lf|  \meanlr{(\mathcal{H}_{\eps} + \xi)^{-1} \psi}{i \nabla_j \cdot \lf(\aav_{j,\eps}\lf(\xv_j\ri) -\aav_{j,\mu}\lf(\xv_j\ri) \ri) }{(\mathcal{H}_{\text{eff}} + \xi)^{-1} \phi} \ri|.
		\edm
		By Schwartz  inequality and \cref{lemma:1}, {\it i.e.}, the fact that $ \nabla_j (\mathcal{H}_{\varepsilon} + \xi)^{-1}\psi \in L^2(\Lambda^N;\mathbb{C}^s)$ for any
$j =1, \ldots, N$, with norm uniformly bounded w.r.t. $\varepsilon$, $j$, and $\psi$, whenever $\lVert \psi \rVert\leq 1$, we obtain
		\begin{equation}
  			\label{aterm}
			\begin{split}
  				I_{\varepsilon}(\aav)\leq C\sup_{\substack{\phi \,\in L^2(\Lambda^N;\mathbb{C}^s)\\ \lf\|\phi \ri\|_2 \leq 1}} \lf\| \lf(\aav_{j,\eps} -\aav_{j,\mu} \ri) (\mathcal{H}_{\text{eff}} + \xi)^{-1} \phi\,\ri\|_{L^2(\Lambda^N;\mathbb{C}^{s})}^{}\; .
			\end{split}
		\end{equation}
		Fix now $R>0$, and set
		\begin{equation*}
  			U_R = \lf\{ \xxv \in \Lambda^N \: \big| \:  |\xxv| \leq R \ri\}.
		\end{equation*}
		Splitting $\Lambda^N =U_R \cup U_R^{\mathrm{c}}$, we can use Hölder inequality to obtain
  		\bmln{
     		 	\lf\| \lf(\aav_{j,\eps} -\aav_{j,\mu} \ri) (\mathcal{H}_{\text{eff}} + \xi)^{-1} \phi\,\ri\|^2_{2} \leq s \int_{\Lambda^N}^{} \diff \xxv \:  \lf| \aav_{j,\eps}(\xv_j) - \aav_{j,\mu}(\xv_j) \ri|^2 \lf| (\mathcal{H}_{\text{eff}} + \xi)^{-1} \phi(\xxv)  \ri|^2	\\
      	\leq s \lf( \int_{|\xv| \leq R}^{} \diff \xv \: \lf| \aav_{j,\eps}(\xv) - \aav_{j,\mu}(\xv) \ri|^{2p}  \ri)^{1/p} \lf\| (\mathcal{H}_{\mathrm{eff}}+\xi)^{-1}\phi  \ri\|_{2p'}^2\\ + s \lf\| \aav_{j,\eps} - \aav_{j,\mu}  \ri\|_{L^{\infty}(\Lambda;\mathbb{R}^{d})}^2 \int_{U_{R}^{\mathrm{c}}}^{} \diff \xxv \: \lf| (\mathcal{H}_{\mathrm{eff}}+\xi)^{-1}\phi(\xxv) \ri|^2 ,
    		}
  		for any $\frac{1}{p}+\frac{1}{p'}=1$. Now, we would like that $1 \leq p'< \frac{dN}{dN-2} $, for $ dN \geq 2 $ (in order to use Sobolev embedding), and that in addition $p<\infty$, {\it i.e.}, $p'>1$. By taking
  		\begin{equation*}
      		p' = \begin{cases}
                	2,	&	\mbox{if } dN \leq 4, \\ 
               		\frac{dN}{dN-2}, &	\mbox{if } dN >4,
              	\end{cases}
		\end{equation*}
     		all the requests are met, and thus we get
      	\begin{equation*}
           	p = \begin{cases}
                	2,	&	\mbox{if } dN \leq 4, \\ 
               		\frac{1}{2} dN, &	\mbox{if } dN >4.
              	\end{cases}
      	\end{equation*}
     		 In addition, let us recall that by \eqref{eq:1}
 		\begin{equation*}
    			\lf\| \aav_{j,\eps} - \aav_{j,\mu}  \ri\|_{L^{\infty}(\Lambda;\mathbb{R}^{d})} \leq C < +\infty\;, 
  		\end{equation*}
  		uniformly w.r.t. $\varepsilon\in (0,1)$. Therefore, by Sobolev embedding and \cref{lemma:1},
  		\bmln{
      	I_{\varepsilon}(\aav)\leq C\sup_{\substack{\phi \,\in L^2(\Lambda;\mathbb{C}^s)\\ \lf\|\phi \ri\|_2 \leq 1}} \bigg\{ \bigg( \int_{\lf| \xv \ri| \leq R}^{} \diff \xv \: \lf| \aav_{j,\eps}(\xv) - \aav_{j,\mu}(\xv) \ri|^{2p}  \bigg)^{1/p} \lf\| (\mathcal{H}_{\mathrm{eff}}+\xi)^{-1}\phi  \ri\|_{2p'}\\ +   \int_{U_{R}^{\mathrm{c}}}^{} \diff \xxv \: \lf| (\mathcal{H}_{\mathrm{eff}}+\xi)^{-1}\phi(\xxv) \ri|^2 \bigg\}	\\
      	C\sup_{\substack{\phi \,\in L^2(\Lambda;\mathbb{C}^s)\\ \lf\|\phi \ri\|_2 \leq 1}} \bigg\{ \bigg( \int_{\lf| \xv \ri| \leq R}^{} \diff \xv \: \lf| \aav_{j,\eps}(\xv) - \aav_{j,\mu}(\xv) \ri|^{2p}  \bigg)^{1/p} \lf\| (\mathcal{H}_{\mathrm{eff}}+\xi)^{-1}\phi  \ri\|_{H^r(\Lambda^N;\C^s)}\\ +   \int_{U_{R}^{\mathrm{c}}}^{} \diff \xxv \: \lf| (\mathcal{H}_{\mathrm{eff}}+\xi)^{-1}\phi(\xxv) \ri|^2 \bigg\}.
    		}	
  		On the r.h.s., the first term converges to zero as $\varepsilon\to 0$ by dominated
  convergence, since $\aav_{j,\eps}(\xv) - \aav_{j,\mu}(\xv)$ converges to zero $\xv\in \Lambda$-a.e. by \cref{point}. The second term converges to
  zero as $R\to \infty$, because $U^{\mathrm{c}}_{R}=\varnothing$ for any $R $ large enough, when $\Lambda$ is bounded, and, when $\Lambda$  is unbounded, it converges
  to zero because $(\mathcal{H}_{\mathrm{eff}}+\xi)^{-1}\phi\in L^2(\Lambda^N;\mathbb{C}^s)$, with norm uniformly bounded thanks to the condition
  $\lVert \phi \rVert_{2}^{}\leq 1$. Taking first the limit $ \eps \to 0 $ and then $ R \to \infty $, we get the result.
	\end{proof}

\subsection{Rougher electromagnetic fields}
	\label{sec:rough-electr-fields}

As discussed in \cref{rem:2}, the effective potentials $\aav_{j,\mu}$ and fields $\mathbf{B}_{j,\mu}$,
$\mu\in \mathscr{M}\bigl(\homega\bigr)$, as well as the variance electric potentials
$W_{j,\mu} $ obtained in \cref{thm:conveff} via the quasi-classical limit are all bounded (actually
continuous and vanishing at infinity). This is due to the regularity assumption~\eqref{omegacontrol} that we made on the family of Fock quantum vectors
$(\Psi_{\varepsilon})_{\varepsilon\in (0,1)}$ . It is therefore interesting to figure out whether it is
possible to obtain MSOs with less regular $\aav_{j,\mu}$, $\mathbf{B}_{j,\mu}$ and $W_{j,\mu}$ relaxing \eqref{omegacontrol}.

Thanks to \cref{rem: positive}, in order to ensure that $\mathcal{H}_{\mathrm{eff}}(\mu)$ to be self-adjoint, we need to request that (see, \emph{e.g.}, \cite{MR606167,Matte:2017ab})
\begin{equation*}
  \aav_{j,\mu} \in L^2_{\mathrm{loc}}(\Lambda;\mathbb{R}^{d}),	\qquad	 	W_{j,\mu} \in L^1_{\mathrm{loc}}(\Lambda),
\end{equation*}
for any $ j = 1, \ldots, N $, and $ -\vec{\sigma}_j \cdot \mathbf{B}_{j,\mu} $ combined with the negative part of $ V $ is form-bounded w.r.t. the Laplacian, i.e.,
\begin{equation*}
 	\sum_{j = 1}^N \lf\| \vec{\sigma}_j \cdot \mathbf{B}_{j,\mu} \ri\|_{\mathbb{C}^s}^{}+ V_- \in \mathfrak{K}_{<}(\Lambda)\;,
\end{equation*}
where $ V_- = - \min\{V, 0 \} $.

For the sake of simplicity, let us fix $s=1$ in the following discussion ($s=1$ corresponds to spinless particles, hence there is no
Zeeman term in the effective Hamiltonian). Now, removing Assumption~\ref{omegacontrol} altogether we still have convergence, up to an
eventual subsequence extraction, of $\Psi_{\varepsilon} $ to a \emph{cylindrical measure} $M$ on $ L^2(\mathcal{K},\nu) $, \emph{i.e.}, to a Radon
probability measure $\mu_M$ on the Hausdorff completion $\overline{\mathfrak{H}}=\overline{L^2(\mathcal{K},\nu)_{\mathrm{weak}}}$ of $L^2$ endowed
with the weak topology. In addition, all the cylindrical measures are reached by suitable sequences of quantum states
\cite{Falconi:2016ab}. Furthermore, since the symbols involved in the limit are cylindrical (finite rank), as discussed in the proof of \cref{point}, it is possible
to prove pointwise convergence for generic cylindrical measures, provided $ \mathfrak{a}_{(\: \cdot \:)}(\xv_1, \ldots, \xv_N)$ and
$\mathfrak{a}^2_{(\: \cdot \:)}(\xv_1, \ldots, \xv_N)$ are
$M$-integrable as cylindrical functions \cite{Falconi:2017aa
}. Finally, if the measure $\mu_M$ is a linear map from $L^p_{\mathrm{loc}}(\Lambda\times \overline{\mathfrak{H}}; W)$ into
$L^p_{\mathrm{loc}}(\Lambda; W)$ for $p=1,2$ (for any $W$ finite dimensional real vector space), we get
\begin{equation*}
 \aav_{j,\mu_M}= 2 \Re \int_{\overline{\mathfrak{H}}}^{} \mathrm{d}\mu(\zv)  \braket{\zv}{ \lambda_{j}(\xv) \hat{\eev}}_{\mathfrak{H}} \in L^2_{\mathrm{loc}}(\Lambda; \R^d),	\qquad  W_{j,\mu_M} \in L^{1}_{\mathrm{loc}}(\Lambda; \mathbb{R}^+)
\end{equation*}
whenever $ \mathfrak{a}_{\zv}(\xv_1, \ldots, \xv_N)  \in L^2_{\mathrm{loc}}(\Lambda^N \times \overline{\mathfrak{H}}; \mathbb{R}^{dN})$. If these assumptions are satisfied, we can prove the convergence
in strong resolvent sense
\begin{equation*}
  \mathcal{H}_{\eps} \xrightarrow[\eps \rightarrow 0]{\mathrm{s-res}} \mathcal{H}_{\mathrm{eff}} (\mu_M)\; .
\end{equation*}
If $\lambda$ is ``not too regular'', the effective magnetic potential vector $\aav_{j,\mu_M}$ and the variance electric potential
$W_{j,\mu_M} $ may also be non-smooth, \emph{e.g.}, they may belong to
$L^2_{\mathrm{loc}}(\Lambda;\mathbb{R}^{d}) \setminus L^2(\Lambda;\mathbb{R}^{d})$ and $L^{1}_{\mathrm{loc}}(\Lambda;\mathbb{R}^+)\setminus L^{1}(\Lambda;\mathbb{R}^+)$ respectively for any $ j $.

There is an inconvenience hidden in the above discussion, namely that it is difficult to impose explicit conditions on the quantum states that imply that the
aforementioned assumptions on $M$ and $\mu_M$ are satisfied. It is however possible to construct  in a simple way  explicit microscopic models and
corresponding families of quantum Fock states such that \emph{any} given MSO can be obtained in the quasi-classical limit. Let us
discuss the construction in detail for the simple MSO of one spinless particle with no additional external potentials, moving in a
region $\Lambda\subset \mathbb{R}^d$, $d\geq 2$, since the generalization to $N $ particles with an interaction potential $ V $ is trivial: the MSO reads in this case
\begin{equation}
  \label{eq:2}
  \mathcal{H}^{\aav}= \bigl(-i\nabla -\aav(x)\bigr)^2\; ,
\end{equation}
with $\aav\in L^2_{\mathrm{loc}}(\Lambda;\mathbb{R}^d)$ a given (divergence-less) magnetic potential.

	\begin{definition}[Microscopic system]
  		\label{def:1}
  		\mbox{}	\\
  		Let $\mathcal{K}=\hat{\Lambda}$ be the (Pontrjagin) dual of $\Lambda$, with $\nu=h$ Haar measure, and $\mathfrak{H}_{\aav} : = \mathbb{C}^{d-1} \otimes L^2(\hat{\Lambda}, \mathrm{d}\nu)$. Assume also that the particle have the following electromagnetic coupling factor, in Coulomb gauge:
  		\begin{equation*}
    			\vec{\lambda}(\xv;\kv)=\sum_{\gamma=1}^{d-1} \lambda_{\aav}(\kv) e^{-i \kv \cdot \xv} \vec{e}_{\gamma}(\kv),
  		\end{equation*}
  		where $ \eev_{\gamma} $, $ \gamma = 1, \ldots, d-1 $, is the polarization basis, and $\lambda_{\aav}\in \mathfrak{H}_{\aav}$ such that there exists
  $\lambda_{\aav}^{-1}$ satisfying $\lambda_{\aav}(\kv)\lambda_{\aav}^{-1}(\kv)=1$ a.e.. Finally, let $ \omega $ be the dispersion relation of the
  radiation field, {\it i.e.}, a positive function of $ \kv $. Then, the {\it microscopic model} of a spinless particle interacting with the radiation is described by the Hamiltonian
 		\begin{equation}
    			H_{\aav} =  \lf( -i\nabla - \vec{\varphi} (\vec{\lambda}) \ri)^2 + \mathrm{d}\Gamma(\omega).
  		\end{equation}
	\end{definition}

	An explicit example of $\lambda_{\aav}$ is for instance
	\begin{equation*}
  		\lambda_{\aav}(\xv; \kv)= \frac{1}{\sqrt{\lvert \kv  \rvert_{}^{}}} e^{-\frac{\lvert \kv  \rvert_{}^2}{2}}\;,
	\end{equation*}
	corresponding to a Gaussian charge distribution. 
	
	It remains now to define the appropriate quantum states. We make use of squeezed coherent states converging to a measure concentrated outside of $ \C^{d-1} \otimes L^2(\hat{\Lambda};|\kv| \mathrm{d}h(\kv))$. 
	
	\begin{definition}[Squeezed coherent states]
		\label{def:squeezed}
		\mbox{}	\\
		A {\it squeezed coherent state} in $ \Gamma_{\mathrm{s}}(\mathfrak{H}_{\aav})$ is defined as
		\begin{equation}
  			\Xi(f_{1,\varepsilon},\dotsc,f_{\gamma,\varepsilon}) = \exp \lf\{ \tx\frac{1}{\varepsilon} \sum_{\gamma} \lf( a^{\dagger}_{\gamma}(f_{\gamma,\varepsilon})-a_{\gamma}(f_{\gamma,\varepsilon} \ri) \ri\} \Omega\;,
		\end{equation}
		where $ f_{1,\varepsilon},\dotsc, f_{d-1,\varepsilon}  \in  L^2(\hat{\Lambda},\mathrm{d}h)$ and $\Omega\in \Gamma_{\mathrm{s}}(\mathfrak{H}_{\aav})$ is the Fock vacuum vector. 
	\end{definition}
	
Let us also remark that the Fourier transform of the divergence-less
field $\aav\in \mathscr{S}'(\Lambda;\mathbb{R}^d)$ can be written as
\begin{equation}
  \hat{\aav}(\kv)= \sum_{\gamma=1}^{d-1} \hat{A}_{\gamma}(\kv) \vec{e}_{\gamma}(\kv)\; ,
\end{equation}
for some suitable $\hat{A}_{\gamma}\in \mathscr{S}'(\hat{\Lambda})$.

\begin{definition}[Quasi-classical coherent Fock vectors]
  		\label{def:2}
  		\mbox{}	\\
  		Let $ \lf\{ \xi_{\varepsilon} \ri\}_{\varepsilon\in (0,1)}\subset C_0^{\infty}(\hat{\Lambda})$ be a family of standard mollifiers, {\it i.e.}, smooth functions with compact support $ \subset \hat \Lambda $). Then, the {\it quasi-classical family}
  $ \lf\{ \Psi_{\varepsilon}^{\aav} \ri\}_{\varepsilon\in (0,1)}$ of coherent Fock vectors associated to $\aav$ is given by
  		\begin{equation}
  			\label{eq:field states}
    			\Psi_{\varepsilon}^{\aav}=\Xi \lf(\lambda_{\aav}^{-1}(\hat{A}_1*\xi_{\varepsilon}), \dotsc ,\lambda_{\aav}^{-1}(\hat{A}_{d-1}*\xi_{\varepsilon}) \ri)\; .
  		\end{equation}
	\end{definition}
	
	\begin{proof}[Proof of \cref{thm:sconv}]
          The proof exploits the fact that on operators that are polynomial functions of the creation and annihilation operators, the
          Weyl operators act explicitly as translations. The strategy is explained in detail for the quasi-classical limit of
          Nelson-type Hamiltonians in \cite[\textsection 3.4]{Correggi:2017aa}. In addition, due to the analogue of \cite[Proposition 3.11]{Correggi:2017aa}, the classical limit on squeezed coherent states takes a simple and explicit form. To sum up,
            using explicit computations it is possible to prove the following convergence, for any quasi-classical
          family of Fock states $ \Psi_{\eps}^{\mathbf{A}}
            $:
		\begin{equation*}
  			\meanlr{\Psi_{\varepsilon}^{\aav}}{H_{\aav}}{\Psi_{\varepsilon}^{\aav}}_{\Gamma_{\mathrm{s}}(\mathfrak{H})} - c_{\varepsilon} =:\mathcal{H}_{\varepsilon}^{\aav} \xrightarrow[\eps \rightarrow 0]{\text{s-res}} \mathcal{H}^{\aav}\; 
		\end{equation*}
		where $ H_{\aav} $ is the Hamiltonian of the microscopic model provided in \cref{def:1}. This is the desired statement for a single spinless particle with no external potential. The generalization of the above result to $ N $ particles in presence of a many-body potential $ V $ is trivial.
	\end{proof}

\section{Ground State Energy}
\label{sec:ground state}

In order to prove \cref{thm:2}, we state some preparatory lemmas. First of all, since we make use of the
diamagnetic inequality (see, \emph{e.g.}, \cite[Theorem 7.21]{MR1817225}), we recall here explicitly this very well-known result for the convenience of the reader.

	\begin{proposition}[Diamagnetic inequality]
  		\label{prop:2}
  		\mbox{}	\\
  		Let $\Lambda\subset \mathbb{R}^d$ be open and let $\aav \in L^2_{\mathrm{loc}}( \Lambda;\RR^d)$ and $ f \in L^2(\Lambda)$ be such that
  $(\partial_j - i A_j)f \in L^2(\Lambda)$, $j = 1, \ldots, d$. Then, $ |f| \in H^1(\Lambda)$ and pointwise for a.e. $ \xv \in \Lambda$
		\beq
			\bigl|\nabla |f| (\xv) \bigr|    \leq  \bigl| (\nabla - i\aav)f (\xv) \bigr|.  
		\eeq
	\end{proposition}

The first question to address is the boundedness from below of $H_{\varepsilon}$, of $\mathcal{H}_{\eps}$, and of $\inf_{\mu\in \mathscr{M}_{2,\omega}} \Bigl(\underline{\sigma}\bigl(\mathcal{H}_{\text{eff}}(\mu)\bigr) + c(\mu)\Bigr)$.

	\begin{lemma}[Boundedness from below]
  		\label{lemma:2}
  		\mbox{}	\\
		Let the assumptions of \cref{thm:2} be satisfied, then uniformly w.r.t. $\varepsilon\in (0,1) $
		\begin{gather}
  			\label{eq:3} \underline{\sigma}(H_{\varepsilon}) \geq - C > -\infty \; ,    \\
  			\label{eq:4} \underline{\sigma}(\mathcal{H}_{\varepsilon})   \geq - C > -\infty \; ,    \\
  			\label{eq:5} \inf_{\mu \in \mathscr{M}_{2,\omega}}\Bigl(\underline{\sigma}\bigl(\mathcal{H}_{\text{eff}}(\mu)\bigr) + c(\mu)\Bigr)  \geq - C > -\infty \; .
		\end{gather}
	\end{lemma}

	\begin{proof}
  		The bound \eqref{eq:3} follows directly from \cref{selfadthm}.
  
  		In order to prove \eqref{eq:4}, we pick a function $ \psi \in \mathscr{D}[-\Delta+ V_+] $; then, by the diamagnetic inequality (\cref{prop:2}), and positivity of $\mathrm{d}\Gamma(\omega)$ and $V_+$, we get
		\bmln{
  			\meanlr{\psi}{\mathcal{H}_{\eps}}{\psi}_2  = \sum_{j=1}^N \Big\{ \meanlrlr{\psi}{( -i\nabla_j - \aav_{j,\eps}(\xv_j) )^2}{\psi}_2 \\
  			  + \meanlrlr{\psi}{N^{-1} \langle\vec\varphi^2\rangle_{\eps}(\xv_1, \ldots, \xv_N) - \aav^2_{j,\eps}(\xv_j)}{\psi}_2- \meanlrlr{\psi}{\vec{\sigma}_j\cdot \mathbf{B}_{j,\varepsilon}(\xv_j)}{\psi}_2 \Big\} + \meanlrlr{\psi}{V_{<}}{\psi}_2   \\
			 \geq \sum_{j=1}^N \lf\{ \meanlrlr{\psi}{N^{-1} \langle\vec\varphi^2\rangle_{\eps}(\xv_1, \ldots, \xv_N) - \aav^2_{j,\eps}(\xv_j)}{\psi}_2 - \meanlrlr{\psi}{\vec{\sigma}_j\cdot \mathbf{B}_{j,\varepsilon}(\xv_j)}{\psi}_2 \ri\} \\
			 + \meanlrlr{\psi}{-\Delta+ V_{<}}{\psi}_2.
		} 
		Now, $\xv$-a.e. $N^{-1} \langle\vec\varphi^2\rangle_{\eps} - \aav^2_{j,\eps}(\xv)$ is a variance term (see \cref{rem: positive}), and therefore positive. Using then the Kato-smallness of
$V_{<}$ and of $ \sum \vec{\sigma}_j \cdot \mathbf{B}_{j,\varepsilon}(\xv)$, which is actually infinitesimally Kato-small in the sense of quadratic forms, we conclude the proof. 

		It remains to prove
\eqref{eq:5}. Let $\mu \in \mathscr{M}_{2,\omega}$ and $\psi \in \mathscr{D}[-\Delta+ V_+]$. Similarly to the proof of \eqref{eq:4} above, by diamagnetic inequality and
positivity, we get
		\begin{align*}
  			\meanlr{\psi}{\mathcal{H}_{\text{eff}} (\mu)}{\psi}_2  &\geq \sum_{j=1}^N  \meanlrlr{\psi}{( -i\nabla_j - \aav_{j,\mu}(\xv_j) )^2- \vec{\sigma}_j\cdot \mathbf{B}_{j,\mu}(\xv_j)}{\psi}_2 + \meanlrlr{\psi}{V_{<}}{\psi}_2 \\
			& \geq \meanlrlr{\psi}{-\Delta+ V_{<}}{\psi}_2 - \sum_{j = 1}^N \lf\| \vec{\sigma}_j  \ri\|_{\mathbb{C}^{sd}}^{} \lf\| \mathbf{B}_{j,\mu}  \ri\|_{L^{\infty}(\Lambda;\mathbb{R}^{d})}^{}  \; .
		\end{align*}
		The first term of the r.h.s. is bounded from below, since $V_{<}$ is Kato-small by hypothesis. It remains to
show that 
		\bdm
			- \sum_{j = 1}^N \lf\| \vec{\sigma}_j  \ri\|_{\mathbb{C}^{sd}}^{} \lf\| \mathbf{B}_{j,\mu}  \ri\|_{L^{\infty}(\Lambda;\mathbb{R}^{d})}^{} + c(\mu)
		\edm
		is bounded from below uniformly with respect to $\mu\in \mathscr{M}_{2,\omega}$: however,  for any $\delta_j >0$,
		\bmln{
    			\lf\| \mathbf{B}_{j,\mu} \ri\|_{L^{\infty}(\Lambda;\mathbb{R}^{d})}^{}\leq 2 \lf\| \omega^{-1/2} b_j  \hat{\eev} \ri\|_{L^{\infty}(\Lambda;\mathfrak{H}^{d})}^{} \int_{\homega}^{}   \mathrm{d}\mu(z) \: \lf\| z  \ri\|_{\homega}^{}\\\leq \frac{1}{\delta_j} \int_{\homega}^{}   \mathrm{d}\mu(z) \: \lf\| z  \ri\|_{\homega}^{2} + \delta_j (d-1) \lf\| \omega^{-1/2}  b_j   \ri\|_{L^{\infty}(\Lambda)}^{2} \\
    			=\frac{1}{\delta_j} c(\mu)  + \delta_j (d-1) \lf\| \omega^{-1/2}  b_j   \ri\|_{L^{\infty}(\Lambda)}^{2}.
  		}
		Choosing now $\delta_j = N \lf\| \vec{\sigma}_j  \ri\|_{\mathbb{C}^{sd}}$, we get
		\beq
      		- \sum_{j = 1}^N \lf\| \vec{\sigma}_j  \ri\|_{\mathbb{C}^{sd}}^{} \lf\| \mathbf{B}_{j,\mu}  \ri\|_{L^{\infty}(\Lambda;\mathbb{R}^{d})}^{} + c(\mu) \geq - N (d-1) \sum_{j =1}^N \lf\| \vec{\sigma}_j  \ri\|_{\mathbb{C}^{sd}} \lf\| \omega^{-1/2}  b_j   \ri\|_{L^{\infty}(\Lambda)}^{2}
  		\eeq
  		which is independent of $ \mu $ and bounded from below.
	\end{proof}

We have now all the elements to prove \cref{thm:2}.

	\begin{proof}[Proof of \cref{thm:2}]
  		Let us set $\mathcal{J}_{\text{eff}}(\mu) := \mathcal{H}_{\text{eff}}(\mu) + c(\mu)$. Following \cite[Lemma 3.20]{Correggi:2017aa}, it is possible to prove that
  \[ \inf_{\mu \in \mathscr{M}_{2,\omega}} \underline{\sigma} (\mathcal{J}_{\text{eff}}(\mu)) = \inf_{\mu \in \mathscr{M}_{\text{fin}}} \underline{\sigma}
    (\mathcal{J}_{\text{eff}}(\mu))\; ,\] where $\mathscr{M}_{\text{fin}}:=\mathscr{M}_{\text{fin}}(\homega)$ is the set of finitely
  supported Radon probability measures on $\homega$, \emph{i.e.}, the subset of $\mathscr{M}_{2,\omega}$ consisting of
  measures such that there exists $I\subset \mathbb{N}$ finite, $\{ \alpha_j\}_{j \in I } \subset [0,1]$, with $\sum_{j \in I } \alpha_j = 1$, and
  $\{ \zv_j\}_{j \in I }\subset \homega $, such that
  \[ \mu = \sum_{j \in I} \alpha_j \delta( \zv - \zv_j)\; . \] Let us consider now vectors of the form $\phi\otimes \Xi(\zv)$, with
  $\phi\in \mathcal{D}(-\Delta_{\mathrm{D}}+V_+)$ and $ \zv \in \homega 
  $, and where we recall that $\Xi$ is the squeezed coherent state as in \cref{def:squeezed}: for a vector-valued $ \zv = (z_1, \ldots, z_{d-1}) $, we denote by $ \Xi $ the following tensor product
 		\bdm
 			\Xi(\zv) : = \Xi(z_1) \otimes \cdots \otimes \Xi(z_{d-1}) \in \C^{d-1} \otimes L^2_{\omega}(\mathcal{K}).
		\edm
		Using the properties of squeezed coherent states and, in particular, that
  $\Xi(z)\to \delta(\,\cdot\, -z)$, the upper bound is easy to prove. In fact, setting
  		\bdm
  			\mathcal{F} :=\bigg\{ \{\zv_j\}_{j\in I} \subset \homega 
  			\: \Big| \: I\subset \mathbb{N}\text{ finite}\,,\, \{\alpha_j\}_{j\in I}\subset [0,1],\sum_{j\in I}^{}\alpha_j=1 \bigg\},
		\edm
		we get
  		\bml{
    			\underline{\sigma}(H) \leq \inf_{\mathcal{F}} \;\inf_{\phi \in \mathcal{D}(-\Delta_{\mathrm{D}}+V_+), \lf\| \phi \ri\|_2 = 1}\;  \sum_{j \in I} \alpha_j \meanlr{\phi\otimes \Xi(\zv_j)}{H_{\varepsilon}}{\phi\otimes \Xi(\zv_j)}_{\mathscr{H}} \\
    			= \inf_{\mu \in \mathscr{M}_{\text{fin}}} \;\inf_{\phi \in \mathcal{D}(-\Delta_{\mathrm{D}}+V_+), \lf\| \phi \ri\|_2 = 1}\; \meanlr{\phi}{ \mathcal{J}_{\text{eff}}(\mu)}{\phi}_{L^2(\Lambda)} =  \inf_{\mu \in \mathscr{M}_{\text{fin}}} \underline{\sigma} (\mathcal{J}_{\text{eff}}(\mu)) = \inf_{\mu \in \mathscr{M}_{2,\omega}} \underline{\sigma} (\mathcal{J}_{\text{eff}}(\mu)).
 		}
 		
		It is then sufficient to show that also
  		\begin{equation}
    			\liminf_{\varepsilon\to 0}\underline{\sigma}(H) \geq  \inf_{\mu \in \M_{2,\omega}} \underline{\sigma}\bigl(\mathcal{J}_{\mathrm{eff}}(\mu)\bigr)\; .
  		\end{equation}
  		Since $\mathcal{D}(-\Delta_{\mathrm{D}}+V_+)\otimes \mathcal{D}(\mathrm{d}\Gamma(\omega))$ is a core of $H_{\varepsilon}$, for any $\delta>0$, it is possible to find a vector
  $\Pi_{\varepsilon,\delta}\in \mathcal{D}(-\Delta_{\mathrm{D}}+V_+)\otimes \mathcal{D}(\mathrm{d}\Gamma(\omega))$ satisfying
  		\begin{equation}\label{controlpi}
    			\meanlr{\Pi_{\varepsilon,\delta}}{H_{\eps}}{\Pi_{\varepsilon,\delta}}_{\mathscr{H}}< \underline{\sigma}(H_{\eps})+\delta\; .
  		\end{equation}
  Since in addition the vectors of the form $\psi\otimes \Xi(\zv)$, with $\psi\in \mathcal{D}(-\Delta_{\mathrm{D}}+V_+)$ and
  $\zv \in \homega $, are total in $ \mathscr{H}$ ($ L^2_{\omega}(\mathcal{K})$ is dense in $ L^2(\mathcal{K}, \diff\nu) $ and therefore $ \homega $ is dense in $\mathfrak{H}$), and belong to
  $\mathcal{D}(-\Delta_{\mathrm{D}}+V_+)\otimes \mathscr{D}(\mathrm{d}\Gamma(\omega))$, it is possible to choose for any $\delta>0$, the vector $\Pi_{\varepsilon,\delta}$ in the form
  		\begin{equation*}
    			\Pi_{\varepsilon,\delta}=\sum_{j=1}^{M(\delta)}\zeta_{j,\delta}(\varepsilon)\psi_{j,\delta}\otimes \Xi\bigl(\zv_{j,\delta}\bigr)\; ,
  		\end{equation*}
  		where $\zv_{i,\delta}\neq \zv_{k,\delta}$, for $i\neq k$, and each $\psi_{i,\delta}$ is normalized; finally, the $\zeta_{i,\delta}(\varepsilon)$ satisfy
  		\begin{equation}
    			\label{eq:7}
    			\sum_{j=1}^{M(\delta)}\lvert \zeta_{j,\delta}(\varepsilon)  \rvert_{}^2+2 \Re \sum_{j<k}^{} \zeta^{*}_{j,\delta}(\varepsilon)\zeta_{k,\delta}(\varepsilon) \braketl{\psi^*_{j,\delta}}{\psi_{k,\delta}}_2 \, e^{-\frac{i}{\varepsilon}\Im \braketl{\zv_{j,\delta}}{\zv_{k,\delta}}_{\mathfrak{H}}-\frac{1}{2\varepsilon}\lVert \zv_{j,\delta} - \zv_{k,\delta}  \rVert_{\mathfrak{H}}^{2}}=1\; .
  		\end{equation}
  		Now,
  		\begin{equation*}
    			\lf| 2 \Re \zeta^{*}_{j,\delta}(\varepsilon)\zeta_{k,\delta}(\varepsilon) \braketl{\psi^*_{j,\delta}}{\psi_{k,\delta}}_2 \ri| \leq \lvert \zeta_{j,\delta}  \rvert_{}^2+\lvert \zeta_{k,\delta}  \rvert_{}^2\; ,
  		\end{equation*}
  		so \eqref{eq:7} yields
                $\sum_{j=1}^{M(\delta)}(1-C_{j,\delta}(\varepsilon))\lvert \zeta_{j,\delta}(\varepsilon)  \rvert_{}^2 \leq 1$,
  		where
  		\bdm
  			0\leq \lvert C_{j,\delta}(\varepsilon)\rvert \leq 2M(\delta)\max_{j<k\in M(\delta)} e^{-\frac{1}{2\varepsilon}\lVert \zv_{j,\delta} - \zv_{k,\delta}
    \rVert_{\mathfrak{H}}^{2}},
		\edm
		 and $C_{j,\delta}(\varepsilon)\to 0$ as $\varepsilon\to 0$. Therefore, it follows that each
  $(\zeta_{j,\delta}(\varepsilon))_{\varepsilon\in (0,1)}$ is uniformly bounded for $\varepsilon$ small enough, \emph{e.g.}, for $ \eps = \varepsilon_0 >0$ such that $C_{j,\delta}(\varepsilon_0) \leq 1/2$ for any $j$.

  		Given $\Pi_{\varepsilon,\delta}$ of this form, the corresponding expectation of $H_{\varepsilon}$ can be explicitly computed and takes the form
  		\bml{
    			\label{eq:cohexp}
      			\meanlr{\Pi_{\varepsilon,\delta}}{ H_{\varepsilon}}{\Pi_{\varepsilon,\delta}}_{\mathscr{H}} = \sum_{j,k=1}^{M(\delta)} {\zeta}^*_{j,\delta}(\varepsilon)\zeta_{k,\delta}(\varepsilon) e^{-\frac{i}{\varepsilon}\Im \braketl{\zv_{j,\delta}}{\zv_{k,\delta}}_{\mathfrak{H}}-\frac{1}{2\varepsilon}\lVert \zv_{j,\delta} - \zv_{k,\delta}  \rVert_{\mathfrak{H}}^{2}} \\
      			\times \Bigl\{ \bigl\langle \psi_{j,\delta}  \bigl\lvert -\Delta_{\mathrm{D}} + V(\xv_1, \ldots, \xv_N) \bigr\rvert \psi_{k,\delta} \bigr\rangle_2 + \meanlr{\zv_{j,\delta}}{\diff \Gamma(\omega)}{\zv_{k,\delta}}_{\mathfrak{H}} \\
      			+ i \sum_{\ell =1}^N \lf[ \meanlrlr{\psi_{j,\delta}}{\lf( \braket{\zv_{j,\delta}}{\lambda_\ell(\xv_\ell) \hat{\eev}}_{\mathfrak{H}} \cdot \nabla_\ell + \nabla_\ell \cdot \braketr{\lambda_\ell(\xv_\ell) \hat{\eev}}{\zv_{k,\delta}}_{\mathfrak{H}}\ri)}{\psi_{k,\delta}}_2 \ri.	\\
      			- \meanlrlr{\psi_{j,\delta}}{\lf(\vec{\sigma}_\ell \cdot  \braketr{\zv_{j,\delta}}{b_\ell(\xv_\ell) \hat{\eev}}_{\mathfrak{H}} +\vec{\sigma}_{\ell} \cdot  \braketr{b_\ell(\xv_\ell) \hat{\eev}}{\zv_{k,\delta}}_{\mathfrak{H}} \ri)}{\psi_{k,\delta}}_2	\\
      			 + \meanlrlr{\psi_{j,\delta}}{\lf( \braket{\zv_{j,\delta}}{\lambda_\ell(\xv_\ell) \hat{\eev}}_{\mathfrak{H}} \cdot \braket{\zv_{j,\delta}}{\lambda_\ell(\xv_\ell) \hat{\eev}}_{\mathfrak{H}} + \braket{\lambda_\ell(\xv_\ell) \hat{\eev}}{\zv_{k,\delta}}_{\mathfrak{H}} \cdot \braket{\lambda_\ell(\xv_\ell) \hat{\eev}}{\zv_{k,\delta}}_{\mathfrak{H}} \ri)}{\psi_{k,\delta}}_2	  \\
      			\lf. \lf. + \meanlrlr{\psi_{j,\delta}}{\braket{\zv_{j,\delta}}{\lambda_\ell(\xv_\ell) \hat{\eev}}_{\mathfrak{H}} \cdot \braket{\lambda_\ell(\xv_\ell) \hat{\eev}}{\zv_{k,\delta}}_{\mathfrak{H}} + \eps (d-1) \lf\| \lambda_\ell(\xv_\ell) \ri\|_{\mathfrak{H}}^2 }{\psi_{k,\delta}}_2	\ri] \ri\}.
    			}
  			Now let $(\zeta_{j,\delta})_{j=1}^{M(\delta)}\subset \mathbb{C}$ be cluster points of each
  $(\zeta_{j,\delta}(\varepsilon))_{\varepsilon\in (0,1)}$ corresponding to a common subsequence, satisfying
  		\begin{equation*}
    			\sum_{j=1}^{M(\delta)}\lvert \zeta_{j,\delta}  \rvert_{}^2=1\;,
  		\end{equation*}
  		then, it is not difficult to see that the corresponding cluster point of~\eqref{eq:cohexp} has the form
  		\begin{equation*}
    			\sum_{j=1}^{M(\delta)} \lf| \zeta_{j,\delta}  \ri|_{}^2 \meanlrlr{\psi_{j,\delta}}{ \mathcal{J}_{\mathrm{eff}}\bigl(\delta(\cdot - \zv_{j,\delta})\bigr)}{\psi_{j,\delta}}_2\; .
  		\end{equation*}
  		Therefore, setting
  		\bml{
      			\mathscr{D}_{\mathrm{gs}}(\delta) :=\biggl\{ \lf( (\zeta_{j,\delta})_{j=1}^{M(\delta)}, \;\{ \zv_{j,\delta}\}_{j=1}^{M(\delta)},  \{\psi_{j,\delta}\}_{j=1}^{M(\delta)} \ri) \subset \mathbb{C} \times \homega \times L^2(\Lambda^N; \C^s) \: \Big| \: 	\\
      			\sum_{j=1}^{M(\delta)} \lf| \zeta_{j,\delta}  \ri|_{}^2=1, \lf\| \psi_{j,\delta} \ri\|_2 = 1 \biggr\}\; ,
    		}
  		we have
  		\begin{equation*}
    			\begin{split}
     			\liminf_{\varepsilon\to 0}\;\meanlr{\Pi_{\varepsilon,\delta}}{ H_{\varepsilon}}{\Pi_{\varepsilon,\delta}}_{\mathscr{H}} =\inf_{\mathscr{D}_{\mathrm{gs}}(\delta)} \; \sum_{j=1}^{M(\delta)} \lf| \zeta_{j,\delta} \ri|^2 \meanlrlr{\psi_{j,\delta}}{ \mathcal{J}_{\mathrm{eff}}\bigl(\delta(\cdot - \zv_{j,\delta})\bigr)}{\psi_{j,\delta}}_2	\\
     			\geq \inf_{\lf\| \{ \zeta_{j,\delta} \}_j \ri\|_{\ell^{2}}^{}=1, \;   \zv_{j,\delta} \in  \mathfrak{H}_{\omega}}\; \sum_{j=1}^M \lf| \zeta_{j,\delta}  \ri|_{}^2 \underline{\sigma}\Bigl(\mathcal{J}_{\mathrm{eff}}\bigl(\delta(\cdot -\zv_{j,\delta})\bigr)\Bigr)\\
     			\geq \inf_{\mu\in \M_{\mathrm{fin}}}\underline{\sigma}(\mathcal{J}_{\mathrm{eff}}(\mu))=\inf_{\mu\in \M_{\omega}}\underline{\sigma}(\mathcal{J}_{\mathrm{eff}}(\mu))\; .
    			\end{split}
  		\end{equation*}
  		Hence
  		\begin{equation*}
    			\inf_{\mu\in \M_{\omega}}\underline{\sigma}(\mathcal{J}_{\mathrm{eff}}(\mu))< \underline{\sigma}(H_{\varepsilon})+\delta\; ,
  		\end{equation*}
  		for any $\delta>0$, and this completes the proof.
	\end{proof}

\appendix

\section{Infinite Dimensional Semiclassical Analysis}
\label{sec:Wick-quant}

In this Appendix we briefly outline some of the tools of bosonic semiclassical analysis in infinite dimensions used throughout the text, see, \emph{e.g.}, \cite{ammari:nier:2008,MR2513969,MR2802894,2011arXiv1111.5918A,Ammari:2014aa,Ammari:2017aa,Ammari:2014ab,ammari:15,MR3390788,Falconi:2016ab,Falconi:2017aa} for an overview of the theory and its
applications. We essentially adopt the notation of \cite{ammari:nier:2008
}, and to a lesser extent of \cite{Falconi:2016ab
  ,Falconi:2017aa
}. Throughout the rest of this section, let us denote by $\mathscr{Z}$ a generic separable complex Hilbert space.

	\begin{definition}[Polynomial symbols]
		\label{def:poly}
		\mbox{}	\\
  		Let $p,q \in \mathbb{N}$. A function $ s:\mathscr{Z}\to \mathbb{C}$ is a $(p,q)$-homogeneous polynomial symbol on $\mathscr{Z}$, denoted as $ s \in \mathscr{P}_{p,q}(\mathscr{Z})$, iff there exist
    $\tilde{s} \in \mathscr{L}(\mathscr{Z}^{\otimes_{\mathrm{s}} p}; \mathscr{Z}^{\otimes_{\mathrm{s}} q})$, such that
  		\bdm
  			\begin{cases}
				s(z) = \braketr{z^{\otimes q}}{\tilde{s}\; z^{\otimes p}}_{ \mathscr{Z}^{\otimes_{\mathrm{s}} q}},	\\
				\tilde{s} = \tfrac{1}{p!} \tfrac{1}{q!}\; \partial^p_z\, \partial^q_{\bar{z}}\, s(z),
			\end{cases}
		\edm
		where the derivatives in the second equation are Gâteaux derivatives.
		\newline
		In addition, a $(p,q)$-homogeneous symbol $s$ is compact, denoted as $ s \in \mathscr{P}^{\infty}_{p,q}(\mathscr{Z})$, iff the corresponding
  operator $\tilde{s} \in \mathscr{L}^{\infty}(\mathscr{Z}^{\otimes_{\mathrm{s}} p}; \mathscr{Z}^{\otimes_{\mathrm{s}} q})$ (\emph{i.e.}, it is compact).
	\end{definition}

It is possible to quantize $(p,q)-$homogeneous symbols using the well-known Wick quantization rules. To $ s \in \mathscr{P}_{p,q}(\mathscr{Z})$ it is
associated the operator $s^{\text{Wick}}$
\begin{gather*}
  s^{\mathrm{Wick}}\restriction_{\mathscr{Z}^{\otimes_n}} \in \mathscr{L}(\mathscr{Z}^{\otimes_{\mathrm{s}} n};\mathscr{Z}^{\otimes_{\mathrm{s}} n+q-p})\; ,\\
  s^{\mathrm{Wick}}\restriction_{\mathscr{Z}^{\otimes_n}} := 1_{[p,+\infty)} (n) \frac{\sqrt{n!(n+q-p)!}}{(n-p)!} \eps ^{\frac{p+q}{2}} \left(\tilde{s} \otimes_{\mathrm{s}} \mathds{1}\restriction_{\mathscr{Z}^{\otimes_{\mathrm{s}} n-p}}  \right)\; .
\end{gather*}
The aforementioned construction can be extended to suitable homogeneous polynomial symbols that are only densely defined (\emph{e.g.}, symbols $s$
whose corresponding $\tilde{s}$ is a densely defined closed operator). The basic symbol-operator correspondences heavily used in bosonic quantum field
theories are the following: let $\xi \in \mathscr{Z}$, and $T$ self-adjoint on $\mathscr{Z}$, then
\begin{align}
s_{\xi}(z) = \braket{\xi}{z}_{\mathscr{Z}} \quad &\longleftrightarrow \quad s_{\xi}^{\mathrm{Wick}} = a(\xi); \label{eqn:a}\\
\bar{s}_{\xi}(z) = \braket{z}{\xi}_{\mathscr{Z}} \quad &\longleftrightarrow \quad \bar{s}_{\xi}^{\mathrm{Wick}} = a^{\dagger}(\xi); \label{eqn:a*}\\
s_{T}(z) = \braket{z}{T z}_{\mathscr{Z}} \quad &\longleftrightarrow \quad s_{T}^{\mathrm{Wick}} = \mathrm{d}\Gamma(T). \label{eqn:gammacorr}
\end{align}
More generally (see, \emph{e.g.}, \cite[Proposition 2.4]{ammari:nier:2008}), for any $\eta_j, \xi_k \in \mathscr{Z}$
\begin{equation}\label{corresp}
s(z) = \prod_{j=1}^p \prod_{k=1}^q  \braket{z}{\eta_j}_{\mathscr{Z}} \braket{\xi_k}{z}_{\mathscr{Z}} \quad \longleftrightarrow \quad s^{\text{Wick}} = 
a^\dagger(\eta_1) \ldots a^\dagger(\eta_p)a(\xi_1)\ldots a(\xi_q)\; .
\end{equation}
The reader interested in quantization procedures and semiclassical analysis should refer to textbooks such as
\cite{MR983366
  ,MR2952218
} for finite dimensional phase spaces, and as already mentioned
\cite{ammari:nier:2008
  ,Falconi:2017aa
  ,Falconi:2016ab
} for infinite dimensional phase spaces.

Let us now analyze the convergence of bosonic quantum states to classical cylindrical and Radon measures. It is possible to prove convergence of a
family of regular quantum states of the C*-algebra of canonical commutation relations to a cylindrical measure, in two suitable topologies, under very
mild conditions on the quantum states, \emph{e.g.}, uniform boundedness of their norm w.r.t. the semiclassical parameter. The analysis of such
convergence is carried out in detail in \cite{Falconi:2016ab
}; in the Proposition below we provide a partial result that covers what is needed in this paper.

	\begin{proposition}[Convergence of states]
  		\label{convcyl}
  		\mbox{}	\\
			Let $ \{\Psi_{\eps} \}_{\eps \in (0,1)} \subseteq \Gamma_{\mathrm{s}}(\mathscr{Z})$ be a family of Fock space vectors such that 
                        \[ \sup_{\eps \in (0,1)} \lVert \Psi_{\eps} \rVert_{\Gamma_{\mathrm{s}}(\mathscr{Z})} \leq C<+\infty \; ;\] then the set of cluster points
                        $\mathrm{Cluster}^{\mathfrak{P}}_{\varepsilon\to 0}(\Psi_{\eps},\varepsilon\in (0,1)) $ w.r.t. a suitable topology
                        $\mathfrak{P}$\footnote{$\mathfrak{P}$ is the weak topology on quantum states ($\Psi_{\varepsilon}$ defines uniquely the
                          state $\langle \Psi_{\varepsilon} \,\vert\; \cdot \; \Psi_{\varepsilon} \rangle_{\Gamma_{\mathrm{s}}}$) of the C*-algebra of canonical commutation relations, with
                          respect to compactly supported and smooth cylindrical observables, see \cite{Falconi:2016ab} for additional
                          details.} is not empty. If in addition, $\{\Psi_{\eps} \}_{\eps \in (0,1)}$ satisfies the no loss of mass
                        condition, the convergence holds in the upper bound topology
                        $\mathfrak{P}\vee\mathfrak{T}$\footnote{$\mathfrak{T}$ is the (preimage of the) topology of pointwise convergence
                          of the generating functional
                          $ \braket{\Psi_{\varepsilon}}{ W(\cdot )\Psi_{\varepsilon}}_{\Gamma_{\mathrm{s}}}$ associated to
                          $\Psi_{\varepsilon}$, where $W(\cdot )$ is the Weyl operator
                          \cite{Falconi:2016ab
                          }.}. If $M \in \mathrm{Cluster}^{\mathfrak{P}}_{\varepsilon\to 0}(\Psi_{\eps},\varepsilon\in (0,1))$, then
                        $M \in \mathscr{M}_{\mathrm{cyl}}(\mathscr{Z})$ is a cylindrical measure on $\mathscr{Z}$.
	\end{proposition}

The notation $\Psi_{\eps} \rightarrow M$ and $\Psi_{\varepsilon}\to \mu$, used throughout the text, means convergence (up to an eventual subsequence extraction as explained in Sect. \cref{sec:conveff}) in the upper bound $\mathfrak{P}\vee\mathfrak{T}$ topology, respectively to a cylindrical measure $M$ and to a cylindrical
measure concentrated as a Radon probability measure $\mu$ (on some suitable space).

Adding additional hypotheses, such as assumption \eqref{omegacontrol}, it is possible to obtain additional information on the cluster points of a family of
Fock vectors. In fact, it is possible to prove their concentration as Borel Radon measures on a suitable space, that may differ from $\mathscr{Z}$
(\emph{e.g.},  on the space $\homega $ for vectors satisfying assumption \eqref{omegacontrol}, that is not embedded nor embedding
$\mathfrak{H}$, if $\omega(k)=\lvert k \rvert_{}^{}$, although they share a common dense subset). The following theorem is again an adaptation to our
context (see Sect. \cref{sec:PF} for the definitions of the objects appearing below) of a more general result, proved in \cite[Theorem 3.3]{Falconi:2017aa
}.

	\begin{thm}[Concentration of Wigner measures]
		\label{concstates}
		\mbox{}	\\
  Let $(\Psi_{\eps})_{\varepsilon\in (0,1)}\subset \Gamma_{\mathrm{s}}(\mathfrak{H})$, satisfying \eqref{omegacontrol}. If
  $\Psi_{\varepsilon}\to_{\mathfrak{P}} M\in \mathscr{M}_{\mathrm{cyl}}(\mathfrak{H})$, then the convergence holds as well in the
  $\mathfrak{P}\vee\mathfrak{T}$ topology (introduced in \cref{convcyl}), and $M$ is concentrated as a Borel Radon measure
  $\mu\in \mathscr{M}_{2\delta,\omega}\subset \mathscr{M}\bigl(\homega\bigr)$ of probability measures with finite momenta
  up to order $2\delta$. In particular,
  \begin{equation}
    \label{norm}
    \int_{\homega}^{} \mathrm{d}\mu(\zv) \: \lVert \zv \rVert^{2\delta}_{\homega}\,  \leq C(\delta)\; ,
  \end{equation}
  where $C(\delta)$ is the constant appearing in \eqref{omegacontrol}.
\end{thm}

	Let us conclude with two following results. The the first one is an adaptation of \cite[Corollary
6.4]{ammari:nier:2008} to vectors satisfying \cref{concstates}.

	\begin{proposition}[Classical limit]
  		\label{eqn:comp}
  		\mbox{}	\\
  		Let $(\Psi_{\eps})_{\varepsilon\in (0,1)}\subset \Gamma_{\mathrm{s}}(\mathfrak{H})$, satisfying \eqref{omegacontrol} and such that
  \begin{equation*}
    \Psi_{\varepsilon}\to \mu\in \mathscr{M}_{2\delta,\omega}\; .
  \end{equation*}
  Then, for any compact polynomial symbol of order at most $2\delta$,
  \begin{equation*}
    s\in \bigoplus_{\substack{p,q\in \mathbb{N}\\p+q\leq 2\delta}} \mathscr{P}^{\infty}_{p,q}\bigl(\homega\bigr)\; ,
  \end{equation*}
  the quantum expectation of its Wick quantization converges to the classical expectation of the symbol:
 		 \begin{equation*}
    			\lim_{\varepsilon\to 0}\, \braketr{\Psi_{\varepsilon}}{  s^{\mathrm{Wick}}\, \Psi_{\varepsilon} }_{\Gamma_{\mathrm{s}}(\mathfrak{H})}=\int_{\homega}^{} \mathrm{d}\mu(\zv)\; s(\zv).
  		\end{equation*}
	\end{proposition}

        It is possible to extend the convergence of
        \cref{eqn:comp} to any polynomial symbol (of degree at most $2\delta$) that can be suitably approximated by a pointwise converging sequence of compact symbols. It is sufficient that the Wick quantizations of the approximating symbols converge in average to the quantization of the original symbol \emph{uniformly w.r.t. $\varepsilon$}; and that the approximating symbols can be dominated by a $\mu$-integrable function. The latter properties are satisfied for example by the free photon energy $\mathrm{d}\Gamma(\omega)=\bigl(\lVert \zv \rVert_{\homega}^2\bigr)^{\mathrm{Wick}}$, that thus converges in average in the limit
        $\varepsilon\to 0$:

	\begin{corollary}[Energy convergence]
		\label{lem:ceps}
		\mbox{}	\\
		Let $ \ceps $ be given by \eqref{eq:ceps}, where $ \Psi_{\eps} $ satisfies \eqref{omegacontrol}, and $ c(\mu) $ by \eqref{eq:cmu}, then,
		\beq
			\ceps \xrightarrow[\eps \to 0]{} c(\mu),
		\eeq
		up to a subsequence extraction.
	\end{corollary}

\section{Pauli-Fierz Operator}
\label{sec:PFO}

We discuss here briefly the question of self-adjointness of the PF Hamiltonian defined in \eqref{eqn:ham}, under the assumptions \eqref{coulomb} and \eqref{eqn:h3}. The functions $ \vec{\lambda}$ and $ \mathbf{b} $ are given in \eqref{eq:lambda} and \eqref{eq:b2} respectively. With such hypothesis, it is possible to prove only essential self-adjointness of $ H $, while the precise domain of self-adjointness can be
characterized with additional assumptions on $\vec{\lambda}$ and $ \mathbf{b} $ \cite{Falconi:2014aa,MR1773809,MR1891842,MR2436496,Matte:2017ab}. 

Let us denote by $\mathscr{D}(T)$ the domain of self-adjointness of an operator $T$, and by $\mathscr{D}[T]$ its form domain. Let us remark that
\cref{selfadthm} below, as well as the other results in this paper, hold also if \eqref{coulomb} is not satisfied, \emph{i.e.}, a different
gauge is chosen. However, with this choice, the Hamiltonian $H$ in \eqref{eqn:ham} can equivalently be written as
\begin{equation}
 H = -\Delta_{\mathrm{D}} + \vec\varphi(\vec{\lambda}) \cdot \vec\varphi(\vec{\lambda}) + i  a^{\dagger}(\vec{\lambda})\cdot \nabla + i \nabla\cdot a(\vec{\lambda})  + V + \mathrm{d}\Gamma(\omega) - \vec{\sigma} \cdot \varphi(\mathbf{b})\; ,
\end{equation}
and $-\Delta_{\mathrm{D}}$ is the Laplacian with \emph{Dirichlet boundary conditions}. All the terms of the operator but the last one are meant to be multiplied by $ \one_{\C^s} $.

	\begin{thm}[Self-adjointness of $ H $]
		\label{selfadthm}
		\mbox{}	\\
  		Let $\vec{\lambda}, \mathbf{b} \in L^{\infty}(\Lambda^N;\mathfrak{H}^{dN})$, with $\nabla\cdot \vec{\lambda}\in L^{\infty}(\Lambda^N;\mathfrak{H})$, and let
 \eqref{eqn:h3} be satisfied. Then,
  		\begin{itemize}\setlength{\itemsep}{2.5mm}
  			\item $H$ is essentially self-adjoint on $\bigl(\mathscr{D}(-\Delta_{\mathrm{D}}+V_{+})\otimes \mathscr{D}(\mathrm{d}\Gamma(\omega))\bigr)\cap C_0^{\infty}(\mathrm{d}\Gamma(1))$;
    
  			\item If, in addition, $ \bm{\lambda} \in L^{\infty}\lf(\Lambda^N; \lf(\mathscr{D}[\omega + \omega^{-1}]\ri)^{dN} \ri) $,
    			$\nabla\cdot \vec{\lambda}\in L^{\infty}\bigl(\Lambda^N;\mathscr{D}[\omega^{-1}]\bigr)$, and
    			$ \mathbf{b} \in L^{\infty}\bigl(\Lambda^N; \mathscr{D}(\omega^{-1})^{dN}\bigr)$, then $H$ is self-adjoint on
    			$\mathscr{D}(-\Delta_{\mathrm{D}}+V_+)\otimes \mathscr{D}(\mathrm{d}\Gamma(\omega))$ and bounded from below, with bound uniform w.r.t. $\varepsilon$.
  		\end{itemize}
	\end{thm}
	
	\begin{proof}
          Essential self-adjointness is proved using the criterion \cite[Theorem 3.1, see also \textsection4.3 for the application to PF-type
          Hamiltonians]{Falconi:2014aa
          }. As already mentioned, there are several different proofs of self-adjointness, with different assumptions
          \cite{MR1891842,MR2436496,Matte:2017ab}, here we have used the most general ones \cite[Theorem
          5.7]{Matte:2017ab
          }.
	\end{proof}

{\footnotesize

\providecommand{\bysame}{\leavevmode\hbox to3em{\hrulefill}\thinspace}
\providecommand{\MR}{\relax\ifhmode\unskip\space\fi MR }
\providecommand{\MRhref}[2]{%
  \href{http://www.ams.org/mathscinet-getitem?mr=#1}{#2}
}
\providecommand{\href}[2]{#2}
}
\end{document}